\documentclass{article}
\usepackage{arxiv}
\usepackage{graphicx}
\usepackage{setspace}
\usepackage{makecell}
\usepackage{booktabs} 
\usepackage{xcolor}
\usepackage{amsmath}
\usepackage{graphicx}
\usepackage{subfigure}
\usepackage{float}
\usepackage{hyperref}
\usepackage{caption}
\usepackage{multirow}
\usepackage{amssymb} 
\usepackage{url}
\usepackage{tikz}
\usepackage{bbm}
\usepackage[super, sort&compress, numbers]{natbib}
\usepackage{diagbox}
\usepackage{tabularx}
\usepackage{amsfonts}
\usepackage{docmute}
\usepackage{appendix}
\usetikzlibrary{shapes,decorations,arrows,calc,arrows.meta,fit,positioning}
\tikzset{
    auto,node distance =1 cm and 1cm,semithick,
    state/.style ={ellipse, draw, minimum width = 0.7cm},
    point/.style = {circle, draw, inner sep=0.04cm,fill,node contents={}},
    bidirected/.style={Latex-Latex,dashed},
    el/.style = {inner sep=2pt, align=left, sloped}
    state/.style={circle, draw, minimum size=0.2cm},
            rectangle state/.style={rectangle, draw, minimum size=0.2cm},
    double state/.style={rectangle, draw, double, minimum size=0.2cm, double distance=0.5pt}, 
    double arrow/.style={double, line width=1pt, double distance=2pt}
}

\hypersetup{
    colorlinks = true,
    linkcolor = blue,
    citecolor = blue,
    urlcolor  = blue
}
\newtheorem{theorem}{Theorem}
\newtheorem{proof}{Proof}
\newcommand{\indep}{\perp \!\!\! \perp}

\title{Causal inference via propensity scores for case--control studies}

\author{
Yan Liu$^{1}$,
Anita Koushik$^{2,3}$,
Philippe Boileau$^{4,5}$,
Cong Jiang$^{6}$,
Miceline M\'esidor$^{7,1}$, \\[2pt]
\bfseries Claudia Waddingham$^{3}$,
Denis Talbot$^{8,9}$,
and Mireille E. Schnitzer$^{1,4,10}$
\\[12pt]
\small
$^{1}$ \normalsize{\textit{Faculty of Pharmacy, Université de Montréal, QC, Canada}}\\ 
 $^{2}$ \normalsize{\textit{Department of Oncology, McGill University, QC, Canada}}\\ 
$^{3}$ \normalsize{\textit{Centre de recherche de St. Mary, QC, Canada}}\\  
$^{4}$ \normalsize{\textit{Department of Epidemiology, Biostatistics and Occupational Health, McGill University,QC, Canada}}\\ 
$^{5}$ \normalsize{\textit{Research Institute of the McGill University Health Centre, QC, Canada}}\\ 
$^{6}$ \normalsize{\textit{Department of Epidemiology and Biostatistics, SUNY Downstate Health Sciences University, NY, United States}}\\
$^{7}$ \normalsize{\textit{Institut national de la recherche scientifique, Centre Armand-Frappier Santé Biotechnologie, QC, Canada}}\\ 
$^{8}$ \normalsize{\textit{Département de médecine sociale et préventive, Université Laval, QC, Canada}}\\
$^{9}$ \normalsize{\textit{Centre de recherche du CHU de Québec – Université Laval, QC, Canada}}\\
$^{10}$ \normalsize{\textit{Département de médecine sociale et préventive, Université de Montréal, QC, Canada}}\\
\\[6pt]
\small \textbf{Correspondence}: Mireille E Schnitzer, Faculty of Pharmacy, Université de Montréal, Montréal, QC, H3C3J7, Canada\\ Email: mireille.schnitzer@umontreal.ca
}

\begin{document}

\onehalfspacing

\maketitle
\begin{abstract}
Propensity score methods for causal inference are increasingly being used in cohort and experimental designs, but their development and uptake in outcome-dependent sampling schemes, such as case--control studies, remains limited. Case--control studies involve the sampling of individuals with and without an outcome of interest with the goal of estimating the effects of past exposures. When the design is observational, statistical adjustment for confounding bias is necessary. In case--control studies, propensity score models can be fit using control data under the assumption that the controls are representative of the source population with respect to their exposure distribution conditional on covariates (``control exchangeability"). In this paper, we first demonstrate that relative effects, such as causal risk ratios, are estimable under three different case--control design variants using control-fitted propensity scores. We appropriate two existing estimators for these designs: inverse probability of treatment weighting and an efficient and doubly robust estimator. We also introduce a novel two-step propensity score caliper-matching procedure for case--control designs.  We introduce novel diagnostic tools to verify two necessary types of overlap.  We then contrast our estimators using simulated data and apply them to examine the association between regular aspirin use and ovarian cancer risk.
\end{abstract}
\textbf{Key words}: Propensity scores, causal inference, case--control studies, inverse probability of treatment weighting, doubly robust estimation, overlap diagnostics

\section*{Introduction}
Propensity score methods are now widely used in observational epidemiological studies to address confounding when randomization is not feasible. Since their formal introduction by Rosenbaum and Rubin, \cite{rosenbaum1983central} propensity score methods have been developed for cohort and experimental designs through weighting, matching, stratification, adjustment, and doubly robust approaches. \cite{rosenbaum1983central, hernan2006estimating, austin2011introduction, stuart2010matching} Application of propensity score methods aims to balance measured confounders across exposure groups and reduce bias in the estimation of causal effects under identification assumptions. However, outcome-dependent sampling inherent in case--control (CC) designs complicates estimation and application. \cite{maansson2007estimation,rose2008simple} 
Perhaps consequently, causal inference methods for CC data have had relatively limited development and uptake. \cite{mesidor2025use} Despite this, CC studies are convenient and frequently employed in the context of rare outcomes and rapidly evolving settings where prospective studies may be infeasible or more expensive.

CC data are typically analyzed using multivariable logistic regression to adjust for confounding, but this approach is biased when the model is misspecified. 
While past work has demonstrated that nonparametric causal effect estimation is feasible in the CC setting when the outcome prevalence or sampling probabilities are known, \cite{van2008estimation, rose2008simple} these quantities are typically not available for the specific population under study. When selected controls are representative of the source population with respect to the exposure distribution conditional on measured covariates (``control exchangeability"), a condition that is notably satisfied under a rare outcome assumption, propensity scores can be consistently estimated using only the controls in CC settings even when the outcome prevalence or sampling probabilities are unknown. \cite{robins1999choice} The rare outcome assumption is typically applicable because the CC design is usually used in this context. \cite{vanderweele2014invited} Under control exchangeability, propensity score models fit with control data can be incorporated into inverse probability of treatment weighting (IPTW) estimators. \cite{maansson2007estimation} Additionally, a semiparametric efficient doubly robust estimator  \cite{jiangTNDDR} has been developed for the test-negative design (TND), a CC variant in which cases and controls are identified according to diagnostic test results among care-seeking individuals. \cite{sullivan2016theoretical} 

Propensity score matching is an alternative design-based method for causal inference. It is often viewed as a more intuitive approach, particularly for ``non-technical audiences",~\cite{RosenbaumRubin1985} but was not previously adapted to CC studies. In general, individual matching of cases to controls based on covariates in CC studies reduces estimation variance. But this kind of matching does not itself control for confounding bias, and may rather introduce selection bias.~\cite{casecontrolmatchingbiasMansournia2023, mansournia2018case}

In this paper, we revisit past results \cite{robins1999choice, maansson2007estimation} for CC designs showing that propensity scores are identified using control data under a rare-disease assumption, allowing for the estimation of population relative effects, even when the prevalence of the outcome in the population or sampling probabilities are unknown. We introduce the more general notion of ``control exchangeability" for general CC designs to clarify the specific condition under which propensity score estimation is possible. We then present a gamut of relative effect estimation procedures to facilitate the uptake of causal inference methodology in CC data analyses. We appropriate an existing IPTW estimator and a semiparametric-efficient and doubly robust estimator to the general CC setting. We introduce a novel two-stage propensity score matching approach that, unlike standard CC matching, can control for confounding bias, and four diagnostics to verify whether there is sufficient propensity score overlap between treated and untreated controls and between cases and controls, respectively. Finally, we apply this methodology to data from a CC study investigating the effect of regular aspirin use on the incidence of ovarian cancer.\cite{sarr_nonopioid_analgesic}

\section*{Methodology}

We are interested in estimating the effect of a binary exposure or treatment $A$ on an outcome $Y$. We denote the multivariate set of measured confounders as $\boldsymbol{C}$. We will consider two target parameters of interest: the marginal population risk ratio (mRR),
\begin{align*} 
    \psi_{mRR}  := \frac{\mathbb{E}\left[\mathbb{P}(Y=1 \mid A=1,  \boldsymbol{C}  ) \right] }{\mathbb{E}\left[\mathbb{P}(Y=1 \mid A=0,  \boldsymbol{C} ) \right]} ,%\label{mRR}
\end{align*}
and the marginal risk ratio among the treated (mRRT),
\begin{align} 
    &\psi_{mRRT}  := \frac{\mathbb{E}\left[\mathbb{P}(Y=1 \mid A=1,  \boldsymbol{C} ) \mid A=1\right ] }{\mathbb{E}\left[\mathbb{P}(Y=1 \mid A=0,  \boldsymbol{C} )\mid A=1 \right]}=\frac{\mathbb{P}(Y=1 \mid A=1 )}{\mathbb{E}\left[\mathbb{P}(Y=1 \mid A=0,  \boldsymbol{C} )\mid A=1 \right]}. \label{mRRT}
\end{align}
Under typical causal assumptions, these parameters can be interpreted causally. The latter parameter may be of most interest when all treated individuals ($A=1$) have a non-zero probability of having been untreated ($A=0$) but some untreated individuals have zero probability of being treated given the values of their confounders $\boldsymbol{C}$.

\subsection*{Case--control design}
CC studies separately recruit individuals with ($Y=1$) and without ($Y=0$) the study outcome. We define $S_1$ to indicate eligibility for study inclusion as a case and $S_0$ for eligibility as a control. Overall eligibility is denoted $S=S_0 \cup S_1$. 
All cases are eligible, so that $\mathbb{P}(S=1\mid Y=1)=1$ or $Y=1 \Rightarrow S=1$. Controls are eligible up to an investigator-specified sampling fraction $\mathbb{P}(S=1\mid Y=0)$ depending on the desired ratio of cases to controls. 
We define $q = \mathbb{P}(S=1)$ to be the overall probability of eligibility, a generally unknown quantity depending on the proportion of cases in the population as well as the investigator-specified sampling fraction.  The complete data are defined as $(\boldsymbol{C},A,Y,S)$. We sample $N_1$ i.i.d. draws from $(\boldsymbol{C},A, Y)\mathbb{I}(S_1=1)$, and $N_0$ i.i.d. draws from $(\boldsymbol{C},A, Y)\mathbb{I}(S_0=1)$, defining $N=N_1+N_0$ with individual-specific data denoted $(\boldsymbol{C}_i,A_i,Y_i),i=1,...,N$.

In the case of a rare outcome, we may assume that the controls are representative of the source population in terms of the conditional distribution of the exposure. That is, $
S_0 \indep A \mid  \boldsymbol{C} $ approximately holds.\cite{robins1999choice} In past work, we named this the ``control exchangeability" assumption.\cite{jiangTNDDR} The directed acyclic graph (DAG) in Figure \ref{dag} (A) represents the CC design. The rarity of $Y=1$ lets us ignore the link between case ($Y=1$) and control ($Y=0$) status (since controls are essentially the entire population) such that the independence between $S_0$ and $A$ can be interpreted directly from the DAG. 
This assumption directly enables estimation of propensity scores through $\mathbb{P}_{CC}(A=a\mid Y=0,\boldsymbol{C}) = \mathbb{P}(A=a\mid Y=0,S=1,\boldsymbol{C})=\mathbb{P}(A=a\mid S_0=1,\boldsymbol{C})=\mathbb{P}(A=a\mid \boldsymbol{C})$, where $\mathbb{P}_{CC}$ represents the probability induced by CC sampling. Throughout the paper, we use $\mathbb{P}$ and $\mathbb{E}$ with no subscript to represent probability and expectation in the source population. 

Identifiability of relative effects with rare outcome CC data using an IPTW g-formula was discussed by Robins \cite{robins1999choice} and Mansson \cite{maansson2007estimation} and applied by Rose and van der Laan. \cite{rose2008simple, rose2014double} They used the identifiability formula, 
\begin{equation}\mathbb{E}\{\mathbb{P}(Y=1\mid A=a, \boldsymbol{C}) \} = q \times \mathbb{E}_{CC}\left\{Y \mathbb{I}(A=a) \middle/ \mathbb{P}(A=a\mid \boldsymbol{C})\right\} = q \times \mathbb{E}_{CC}\left\{Y \mathbb{I}(A=a) \middle/ \mathbb{P}_{CC}(A=a\mid Y=0, \boldsymbol{C})\right\},\label{IPTWform}
\end{equation}
assuming positivity s.t. $0<\mathbb{P}(A=a\mid \boldsymbol{C})<1$ almost surely.
We give our proof in Appendix A.
When the target parameter is a contrast on the relative scale, such as $\psi_{mRR}$, the constant $q$ cancels out and does not need to be estimated.

For the $\psi_{mRRT}$ parameter of the risk ratio effect among the treated, the numerator is simply the probability of experiencing the outcome $Y=1$ in the treated population. It can be identified up to a constant in the CC sampling design through
    $\mathbb{P}(Y=1\mid A=1) = \mathbb{E}_{CC}(Y A ) q_1 /\mathbb{P}_{CC}(A=1)$ where $q_1=\mathbb{P}(S=1\mid A=1)$ is the probability of inclusion in the treated subpopulation.
Similarly, under a weaker positivity condition, $\mathbb{P}(A=1\mid \boldsymbol{C})<1$ almost surely, the denominator can be identified up to the same constant through the IPTW formula,
 $$
    \mathbb{E}\{\mathbb{P}(Y=1\mid A=0, \boldsymbol{C})\mid A=1 \} = q_1 \times \mathbb{E}_{CC}\left\{Y (1-A) \mathbb{P}(A=1\mid \boldsymbol{C})\middle/ \mathbb{P}(A=0\mid \boldsymbol{C}) \mid A=1\right\}.$$
where, as before, the population propensity scores are equal to propensity scores conditional on the control condition and therefore estimable in the CC sample. The proof is given in Appendix B. Again, the constant $q_1$ cancels out in the IPTW formula of $\psi_{mRRT}$ and does not need to be estimated.

\subsection*{Case-control variants}
Case--cohort designs are similar in that they sample all individuals with the outcome ($Y=1$) but sample controls from the general population (cohort) regardless of their outcome. Using the same notation as before, if the sampling of controls is either completely random or conditional on measured covariates, we once again satisfy the assumption that controls are independent of the exposure conditional on covariates, $
S_0 \indep A \mid  \boldsymbol{C}$. Thus, propensity scores are identified using the control data and the same IPTW formulas hold. 

The TND, another CC variant, is discussed in Appendix C. Figure~\ref{dag} (B) and (C) are DAGs representing assumed case--cohort and TND data-generating models, respectively. 

\begin{figure}[!ht]
\centering
  \tikzset{state/.style={circle, draw, minimum size=0.8cm, inner sep=0}, double state/.style={circle, draw, double, minimum size=0.8cm, inner sep=0, line width=1pt}, double arrow/.style={double, -{Latex}, line width=0.5pt, double distance=1pt} }
\begin{subfigure}
    \centering
    \begin{tikzpicture}[remember picture]
         \node[state] (A) {$A$};
         \node[state, minimum size=1.2cm] (Y1) [right=of A] {$Y=1$};
         \node[rectangle state] (S1) [right=of Y1] {$S_1$};
         \node[state, minimum size=1.2cm] (Y0) [below=of Y1] {$Y=0$};
         \node[rectangle state] (S0) [right=of Y0] {$S_0$};
         \node[state] (C) [above left=of A, yshift=.1cm, xshift=-.4cm] {$\boldsymbol{C}$};
         \node (S) [above right=of S0, yshift=-.5cm, xshift=.5cm] { $S=S_1\cup S_0$};
        \draw[>=stealth, ->, line width=0.6pt, scale=1] (A) -- (Y1);
        \draw[>=stealth, ->, line width=0.6pt, scale=1] (C) -- (A);
        \path[>=stealth, ->, line width=0.6pt, scale =1]  (C) edge[bend left=30] (Y1);
        \draw[>=stealth, ->, line width=0.6pt, scale=1] (Y0) -- (S0);
        \draw[double arrow] (Y1) -- (S1);
        \draw[double arrow] (Y1) -- (Y0);
        \node[anchor=north west, overlay, xshift=-10mm, yshift=5mm]at (current bounding box.north west) {(A)};
    \end{tikzpicture}
\end{subfigure}

\vspace{0.6cm}
\begin{subfigure}
    \centering
    \begin{tikzpicture}[remember picture]
        \node[state] (A) {$A$};
        \node[state, minimum size=1.2cm] (Y1) [right=of A] {$Y=1$};
        \node[rectangle state] (S1) [right=of Y1] {$S_1$};
        \node[rectangle state] (S0) [below=of S1] {$S_0$};
        \node[state] (C) [above left=of A, yshift=.1cm, xshift=-.4cm] {$\boldsymbol{C}$};
        \node (S) [above right=of S0, yshift=-.5cm, xshift=.5cm] { $S=S_1\cup S_0$};

       \draw[>=stealth, ->, line width=0.7pt, scale=1] (A) -- (Y1);
       \draw[>=stealth, ->, line width=0.6pt, scale=1] (C) -- (A);
        \path[>=stealth, ->, line width=0.6pt, scale =1] (C) edge[bend left=30] (Y1);
        \draw[double arrow] (Y1) -- (S1);
        \node[anchor=north west, overlay, xshift=-10mm, yshift=5mm]
    at (current bounding box.north west) {(B)};
    \end{tikzpicture}
\end{subfigure}

\vspace{0.6cm}
\begin{subfigure}
\centering
\begin{tikzpicture}[remember picture]
        \node[state] (A) {$A$};
        \node[state] (W1) [right=of A] {$W_1$};
        \node[rectangle state, minimum width=1.2cm, minimum height = 0.6cm] (W) [below=of W1] {$W$};
        \node[rectangle state, minimum width=0.6cm, minimum height = 0.5cm] (S) [right=of W] {$S$};
        \node[state] (W0) [below=of W] {$W_{0}$};
        \node[state] (C) [above left=of A, yshift=.1cm, xshift=-.4cm] {$\boldsymbol{C}$};
        \node (S0) [right=of S, xshift=.5cm] { $S_0=W_0\cap S$};

        \draw[>=stealth, ->, line width=0.6pt, scale=1]  (A) -- (W1);
        \draw[>=stealth, ->, line width=0.6pt, scale=1]  (C)-- (A);
        \path[>=stealth, ->, line width=0.6pt, scale =1]  (C) edge[bend left=30] (W1);
        \path[>=stealth, ->, line width=0.6pt, scale =1]  (C) edge[bend right=30] (W0);
        \path[>=stealth, ->, line width=0.6pt, scale =1]  (C) edge[bend left=45] (S);
        \draw[>=stealth, ->, line width=0.6pt, scale=1]  (W) -- (S);
        \draw[double arrow] (W1) -- (W);
        \draw[double arrow] (W0) -- (W);
        \node[anchor=north west, overlay, xshift=-10mm, yshift=5mm]
    at (current bounding box.north west) {(C)};
        
\end{tikzpicture}

\caption{Directed acyclic graph (DAG) for the case-control design (A), the case-cohort design (B) and the TND (C).~\cite{jiangTNDDR} Boxes indicate control for the variables. Double bar arrows indicate deterministic relationships. 
In designs (A) and (C), our DAG assumptions are not sufficient to imply control exchangeability, $S_0 \indep A \mid \boldsymbol{C}$.}
    \label{dag}
\end{subfigure}
\end{figure}
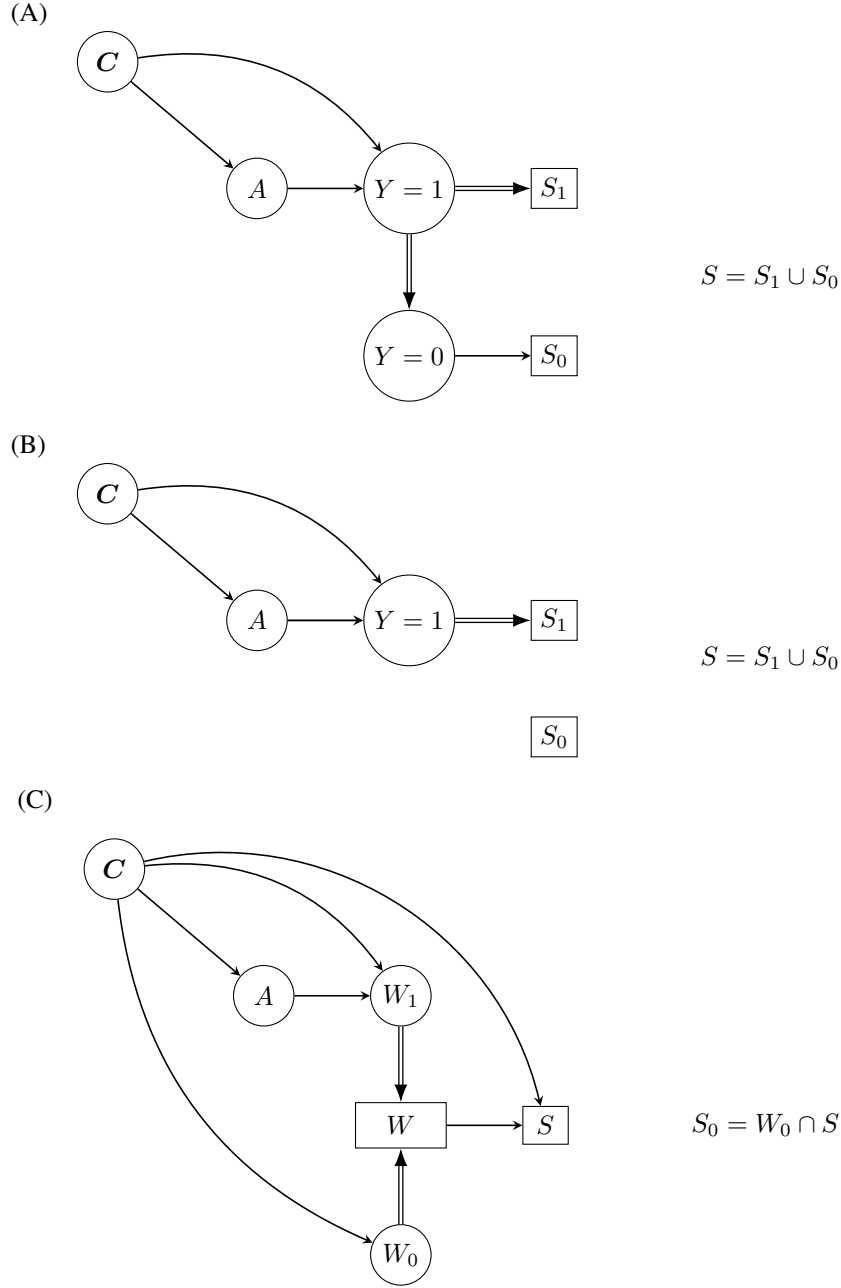

\subsection*{Estimators} 

 For estimation, a propensity score model is fit using only the control data; then this fit is used to predict propensity scores for all study participants. We will review IPTW for CC studies applied  to estimate both $\psi_{mRR}$ and $\psi_{mRRT}$. We also review a semiparametric efficient and doubly robust (or double machine learning) estimator first proposed for the TND \cite{jiangTNDDR} but directly applicable to any of these three CC study designs; we apply this method to estimate  $\psi_{mRR}$. We then propose a novel propensity score matching procedure to estimate $\psi_{mRRT}$.

 \subsubsection*{\textit{Inverse probability of treatment weighted estimation of \texorpdfstring{$\psi_{mRR}$}{psi\_mRR} and \texorpdfstring{$\psi_{mRRT}$}{psi\_mRRT}}} 
 IPTW involves the construction of weights that are used to reweight the sample outcomes in order to adjust for measured covariates. We use $\pi_N(\boldsymbol{C}_i)$ to denote the propensity score $\mathbb{P}(A_i=1\mid \boldsymbol{C}_i)$ for each participant $i=1,...,N$.
 
 For the estimation of $\psi_{mRR}$, the weights for the treated individuals may be defined as  $\omega_{ti}=1/ \pi_N(\boldsymbol{C}_i)$ and for the untreated as  $\omega_{ui}=1/ \{1-\pi_N(\boldsymbol{C}_i)\}$.
 For the estimation of $\psi_{mRRT}$, the weights are redefined for the treated as $\omega_{ti}=1$ and for the untreated as  $\omega_{ui}=\pi_N(\boldsymbol{C}_i)/\{1-\pi_N(\boldsymbol{C}_i)\}$. The IPTW estimator of either parameter  is then given as
\[\frac{\left. \sum_{i=1}^N Y_i A_i \omega_{ti}\middle/ \sum_{i=1}^N A_i \right.} {\left. \sum_{i=1}^N Y_i (1-A_i) \omega_{ui} \middle/ \sum_{i=1}^N A_i \right.}=\frac{\left. \sum_{i=1}^N Y_i A_i \omega_{ti} \right.} {\left. \sum_{i=1}^N Y_i (1-A_i) \omega_{ui}  \right.}.\] 
We denote the estimators of respective quantities as $\hat{\psi}^{IPTW}_{mRR}$ and $\hat{\psi}^{IPTW}_{mRRT}$. We note that if there are some values of ${\pi}_{N}(\boldsymbol{C}_i)$ close to 0 or   1, regularization or truncation may be necessary to avoid numerical instability attributable to divisions by near-zero values.

\subsubsection*{\textit{Doubly robust estimation of \texorpdfstring{$\psi_{mRR}$}{psi\_mRR}}}
Following,\cite{jiangTNDDR} we note that the mRR can be equivalently written as $\psi_{1}/\psi_0$ where $\psi_{a} :=\mathbb{E}\left[\mathbb{P}(Y=1 \mid A=a,  \boldsymbol{C}  = \boldsymbol{c} ) \right]/q$ for $a=(0,1)$.
A one-step efficient and doubly robust estimator of $\psi_{a}$\cite{jiangTNDDR} begins with estimators of $\pi_a(\boldsymbol{C})=\mathbb{P}(A=a\mid \boldsymbol{C})$ and $\mu_a(\boldsymbol{C})=\mathbb{P}_{CC}(Y = 1 \mid A=a ,  \boldsymbol{C} )$, denoted $\pi_{aN}(\boldsymbol{C}_i)$ and $\mu_{aN}(\boldsymbol{C}_i)$ respectively for each participant $i=1,...,N$.  
We then define the one-step estimator of $\psi_{a}$ as
\begin{align}
\hat{\psi}^{OS}_{a}= \frac{1}{N}\sum_{i=1}^N \frac{\mathbb{I}(Y_i = 1, A_i=a)}{{\pi}_{aN}(\boldsymbol{C}_i)} - \frac{{\mu}_{aN}(\boldsymbol{C}_i)}{{\pi}_{aN}(\boldsymbol{C}_i)[1 - {\mu}_{aN}(\boldsymbol{C}_i)]} \mathbb{I}\left( Y_i=0\right) \left[\mathbb{I}(A_i=a) - {\pi}_{aN}(\boldsymbol{C}_i)\right],\quad\text{for } a=0,1. 
\end{align}
The mRR is then estimated as $\hat{\psi}_{mRR}^{OS}=\hat{\psi}^{OS}_{1}/\hat{\psi}^{OS}_{0}$. Cross-fitting can be used to lessen regularity assumptions of the estimators of $\pi_a$ and $\mu_a$. Root-$N$ asymptotic normality of $\hat{\psi}_{a}^{OS}$ requires $N^{-1/4}$ rates of convergence of the estimators of $\pi_a$ and $\mu_a$, which implies that certain flexible machine learning methods can be used to estimate these quantities.~\cite{jiangTNDDR} We note that if there are some values of ${\pi}_{aN}(\boldsymbol{C}_i)$ close to 0 or  ${\mu}_{aN}(\boldsymbol{C}_i)$ close to 1, additional regularization or truncation may be necessary.

\subsubsection*{\textit{Propensity score matching for confounder control for estimation of \texorpdfstring{$\psi_{mRRT}$}{psi\_mRRT}}} 

 We propose a two-stage radius matching procedure with replacement to estimate $\psi_{mRRT}$, which is practical when there are a large number of candidate control participants relative to the treated participants. First, propensity scores are estimated by fitting a model on control participants and predicting scores for the entire sample. Matching is then performed on the logit propensity score, defining closeness between units $i$ and $j$ as the distance $\mid\text{logit}\{\pi_N(\boldsymbol{C}_j)\}-\text{logit}\{\pi_N(\boldsymbol{C}_i)\}\mid$. Both stages utilize radius matching, \cite{heckman1998characterizing,dehejia2002propensity,caliendo2008some} where the first step is restricted to controls and the second matches controls to cases among the untreated individuals, including all available matches within a pre-defined radius (or ``caliper'').

\emph{First stage: matching for covariate balance among the controls.} 
We retain all treated controls (by defining the weight $m_i=1$) and match them with replacement to untreated controls within radius $\rho_1>0$. For untreated controls, $m_i$ is the sum of the fraction importance of each match to a treated observation (e.g., $m_i=1/3+1/4$ if matched to two treated controls that have 3 and 4 total matches, respectively). 
At this stage, one can assess the success of the matching in creating an untreated control group that has good covariate balance with respect to the treated control group. If an acceptable balance is not obtained, matching parameters such as radius and propensity score model can be modified.

\emph{Second stage: donor matching.} 
We now consider the cases. Each treated case $j$ is retained ($m_j=1$). For each untreated case $j$, we identify all untreated controls within a radius $\rho_2>0$ (with replacement) and ``donate'' the mean of their first-stage weights to case $j$. 
So if we identify two untreated controls $i^{'}$ and $i^{''}$ within radius $\rho_2$ of case $j$, we match both to untreated case $j$, by setting $m_j=(m_{i^{'}}+m_{i^{''}})/2$.

The matching estimator can then be defined as
\begin{equation}
\hat{\psi}^{PSM}_{mRRT}=\frac{\left.\sum_{j=1}^N Y_j A_j  \middle/ \sum_{j=1}^N A_j\right.}{\left.\sum_{j=1}^N Y_j (1-A_j) m_j\middle/ \sum_{j=1}^N A_j \right.}
\end{equation}

The proof of convergence of the matching estimator is given in Appendix D. The complete matching algorithm is given in Figure~\ref{ps matching algorithm} and pseudocode in Table S1 in Appendix E. 

Ideally, all treated controls and untreated cases are retained. But if in either stage a match cannot be obtained, the treated control or untreated case is removed, indicating a lack of overlap in the propensity scores between the comparison groups. While radii can be adjusted to limit these exclusions, our implementation in Figure~\ref{ps matching algorithm} forces a nearest-neighbor match in the first stage even if no neighbors are within the radius.

\begin{figure}
    \centering
  \begin{tikzpicture}
            \tikzset{
                state/.style={circle, draw, minimum size=0.8cm, inner sep=0},
                double state/.style={circle, draw, double, minimum size=0.8cm, inner sep=0, line width=1pt},
                double arrow/.style={double, line width=0.5pt, double distance=1pt},
                arrow style/.style={-{Stealth[scale=1]}}
            }
\node (data) at (7,15) [align=center, inner sep=4pt]  
{Data $O=\{(\boldsymbol{C}_i, A_i, Y_i), i = 1, \cdots, N\}$};
  \draw[arrow style] (7, 14.7) -- (7, 14.0) node[below, align=center] {Fit a propensity score model for $\pi_i = \mathbb{P}(A_i=1|\boldsymbol{C}_i)$ in controls ($Y_i=0$)\\make predictions for all participants to obtain estimates $\hat{\pi}_i, \text{for}~ i=1,\cdots, N$.};
\node (split) at (6.5,12.5) [] {Split};
\draw[>=stealth, line width=0.5pt, scale=1] (7, 13) -- (7, 12);
 \draw[>=stealth, line width=0.5pt, scale=1] (1,12) -- (15,12) node(splitline) [midway, above=7pt] {};
    \draw[arrow style] (1, 12) -- (1, 11.5) node[below, align=center] {Treated controls, $S_{1}^{ctrl}$\\$(i:Y_i=0, A_i=1)$};
    \draw[arrow style] (6, 12) -- (6, 11.5) node[below, align=center] {Untreated controls, $S_{0}^{ctrl}$\\$(i:Y_i=0, A_i=0)$};
    \draw[arrow style] (11, 12) -- (11, 11.5) node[below, align=center] {Untreated cases, $S_{0}^{case}$\\$(i:Y_i=1, A_i=0)$};
    \draw[arrow style] (15, 12) -- (15, 11.5) node[below, align=center] {Treated cases, $S_{1}^{case}$\\$(i:Y_i=1, A_i=1)$};

%---------- first stage
     \draw[line width=0.5pt] (1, 10.5) -- (1, 10) node[below, align=center] {};
    \draw[line width=0.5pt] (6.3, 10.5) -- (6.3, 10) node[below, align=center] {};
    \draw[line width=0.5pt] (1,10) -- (6.3,10) node  [midway, above=7pt] {};
    
\draw[arrow style] (3.5, 10) -- (3.5, 8.7) node(sd1)[below, align=center] {Compute $SD\{logit(\hat{\pi}_i):i\in (S_1^{ctrl} \cup S_0^{ctrl} )\}$,\\denote it by $\sigma_1$, then define $\rho_1 = d_1\times \sigma_1$.};
\draw[arrow style] (3.5, 7.8) -- (3.5, 7.2) node(dist1)[below, align=center] {Build distance matrix $D^{ctrl}$ with entries\\$|\text{logit}(\hat{\pi}_j)- \text{logit}(\hat{\pi}_l)|$ for $j\in S_{1}^{ctrl},l\in S_{0}^{ctrl}$.};
\draw[arrow style] (3.5, 6.2) -- (3.5, 5.6) node(match1)[below, align=center] {For each treated control $j\in S_{1}^{ctrl}$, identify\\ all untreated controls within a radius of $\rho_1$\\to define the match set of indices $L_j$.};

\draw[arrow style] (3.5, 4.2) -- (3.5, 3.7) node(eligible1)[below, align=center] {Eligible match exists within $\rho_1$?};
\node (no) at (1.8,2.8) [] {Yes};
\draw[arrow style] (3.4, 3.2) -- (5, 2.3) node(lj1)[below, align=center] {Define \\$L_j = \text{arg~min}_l(D^{ctrl}_{j,l})$ \\for $l\in S_0^{ctrl}$};
\node (no) at (4.8,2.8) [] {No};
\draw[arrow style] (3.4, 3.2) -- (1.5, 2.3) node(lj0)[below, align=center] {Define \\$L_j = \{l\in S_0^{ctrl}$ s.t.\\$D^{ctrl}_{j,l} \leq \rho_1\}$};

\draw[>=stealth, line width=0.5pt, scale=1] (5, 0.9) -- (3.4, 0) node[below, align=center] {};
\draw[>=stealth, line width=0.5pt, scale=1] (1.5, 0.9) -- (3.4, 0) node[below, align=center] {};
\draw[arrow style] (3.4, 0) -- (3.4, -0.5) node(ml)[below, align=center] {Compute the weight for each \\untreated control $m_l=\sum_{j\in S_1^{ctrl}} \mathbbm{1}_{(l\in L_j)}/|L_j|$ };
 \node [draw, dashed, thick, inner sep =0cm, fit =(sd1) (dist1) (match1) (eligible1) (lj1) (lj0) (ml)]{};
 \node(stage1) at (1, 9.1) [] {-- \emph{First stage} --};

%---------- second stage
   \draw[>=stealth, line width=0.5pt, scale=1] (7, 10.5) -- (7, 10) node[below, align=center] {};
    \draw[>=stealth, line width=0.5pt, scale=1] (12.3, 10.5) -- (12.3, 10) node[below, align=center] {};
    \draw[>=stealth, line width=0.5pt, scale=1] (7,10) -- (12.3,10) node  [midway, above=7pt] {};
    
\draw[arrow style] (10.5, 10) -- (10.5, 8.7) node(sd2)[below, align=center] {Compute $SD\{logit(\hat{\pi}_i): i\in (S_0^{case} \cup S_0^{ctrl} )\}$,\\denote it by $\sigma_2$, then define $\rho_2 = d_2\times \sigma_2$.};
\draw[arrow style] (10.5, 7.8) -- (10.5, 7.2) node(dist2)[below, align=center] {Build distance matrix $D^{untrt}$ with entries\\$|\text{logit}(\hat{\pi}_k)- \text{logit}(\hat{\pi}_l)|$ for $k\in S_{0}^{case},l\in S_{0}^{ctrl}$.};
\draw[arrow style] (10.5, 6.2) -- (10.5, 5.6) node(match2)[below, align=center] {For each untreated case $k\in S_{0}^{case}$, identify\\ all untreated controls within a radius of $\rho_2$\\to define the match set of indices $L_k$.};

\draw[arrow style] (10.5, 4.2) -- (10.5, 3.7) node(eligible2)[below, align=center] {Eligible match exists within $\rho_2$?};
\node (no) at (12.2,2.8) [] {No};
\draw[arrow style] (10.5, 3.2) -- (12.3, 2.3) node(lk0)[below, align=center] {Define $L_k = \varnothing$};
\node (no) at (9.2,2.8) [] {Yes};
\draw[arrow style] (10.5, 3.2) -- (9, 2.3) node(lk1)[below, align=center] {Define \\$L_k = \{l\in S_0^{ctrl}$ s.t.\\$D^{untrt}_{k,l} \leq \rho_2\}$};
\draw[arrow style] (12.3, 1.7) -- (12.3, -0.8) node(mk0)[below, align=center] {$m_k = 0$};
\draw[arrow style] (9.2, 0.9) -- (9.2, -0.8) node(mk1)[below, align=left] {$m_k = \frac{1}{\lvert L_k\rvert}\sum_{\ell\in L_k} m_\ell$};

\node (donor) at (6.2,0.2) [] {Donor};
\path (ml) edge[->, bend left=30] (mk1);
\node [draw, dashed, thick, inner sep =0cm, fit =(sd2) (dist2) (match2) (eligible2) (lk1) (lk0) (mk0) (mk1)]{};
\node(stage1) at (8.3, 9.1) [] {-- \emph{Second stage} --};

%-----mRRT
\draw[>=stealth, line width=0.5pt, scale=1] (15, 10.5) -- (15, -3.5) node[below, align=center] {};
\draw[arrow style] (15, -2.3) -- (11.5, -2.3) node[below, align=center] {};
\draw[arrow style] (15, -3.5) -- (10.6, -3.5) node[below, align=center] {};
\node (m1) at (9.3,-2.3) [] {Define $m_t=1$ for $t\in S_1^{case} $};
\draw[arrow style] (7, -1.6) -- (7, -3) node(mrrt)[below, align=left] {Compute the mRRT estimate, \\$\hat{\psi}_{mRRT}^{PSM} =  \sum_{i=1}^N Y_iA_i/ \sum_{i=1}^N Y_i(1-A_i)m_i$};
\node [draw, inner sep =0cm, fit =(mrrt)]{};
\end{tikzpicture}
\caption{Propensity score matching algorithm  for estimating the marginal risk ratio among the treated (mRRT). The matching radii ($\rho_1, \rho_2$) are defined based on the prespecified radius multipliers $d_1$ and $d_2$.}
\label{ps matching algorithm}
\end{figure}
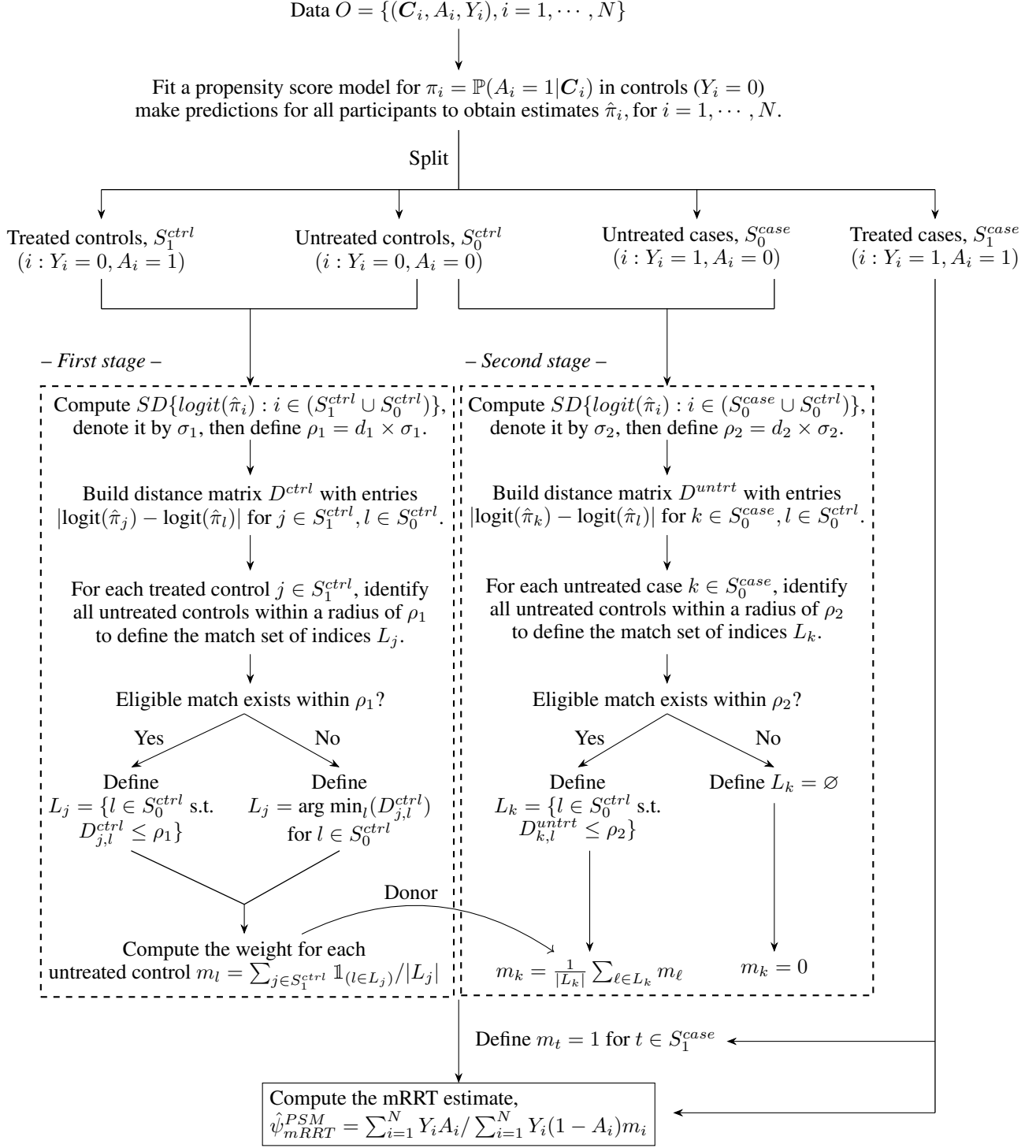

\subsection*{Inference}
Bootstrap resampling to approximate the estimation variance is valid for the IPTW, one-step, and radius matching estimators under regularity assumptions of the estimation of the propensity score function. \cite{rose2008simple, caliendo2008some, austin2014use} 
For CC and case--cohort,  resampling is done separately in the cases and controls while for the TND, resampling is done in the full dataset. 
We use these approaches in the simulation study and application.

\subsection*{Diagnostic procedures}
We propose two steps of CC diagnostic procedures, for the estimation of $\psi_{mRR}$ and $\psi_{mRRT}$. We recommend these procedures for both matching and algorithms using inverse weighting (like IPTW and the one-step estimator).

For both the $\psi_{mRR}$ and $\psi_{mRRT}$, after constructing the weights (either with matching or directly), the first step is to check propensity score overlap and covariate balance between the weighted treated controls and the weighted untreated controls. This step is akin to the standard overlap and balance-checking step in a prospective cohort study since the controls reflect the source population. Propensity score overlap can be visualized using histograms (e.g., top row of Figure \ref{fig: cc_ps}), while standardized mean differences (SMDs) \cite{austin2009balance} can evaluate the similarity in the central tendency of the weighted covariate distributions between treatment groups in the controls; see top half of Table~\ref{tab:bc_cc}.

The second step verifies propensity score overlap between untreated cases and untreated controls (see bottom left of Figure \ref{fig: cc_ps}). When estimating $\psi_{mRR}$, overlap should also be compared between treated cases and treated controls (Section G.3 in Appendix), as non-overlap indicates that we are extrapolating when predicting values for the cases. For matching, we also recommend verifying post-weighting propensity score overlap between untreated cases and untreated ``donor'' controls, where the untreated controls are weighted corresponding to the sum of their donor fractions (see bottom right of Figure \ref{fig: cc_ps}). Covariate balance should be assessed by contrasting each untreated case with its donor set, using the mean difference between the case's covariate value and the donor set mean for each covariate (Table~\ref{tab:bc_cc}).
The complete covariate balance checking algorithms for IPTW and matching are presented in Tables S2 and S3 in Appendix F, respectively.

\section*{Simulations}
We use simulated CC data to contrast the matching estimator and IPTW for the estimation of $\psi_{mRRT}$ and the one-step estimator, IPTW, and logistic regression for the estimation of $\psi_{mRR}$.  We further tested our procedures with simulated TND data (procedures and results are given in Appendix G).

\subsection*{Data simulation}

We repeatedly generated a source population of $N = 2\times 10^6$ individuals with an overall disease prevalence of approximately 1\%, from which CC data were sampled. Two continuous covariates ($C_1$ and $C_2$) were generated from a bivariate normal distribution and affected both exposure ($A$) and disease ($Y$). The outcome $Y$ was generated from a logistic model including $C_1$, $C_2$, $A$, and an interaction between $A$ and $C_1$ to induce effect modification. In scenario (a), the probability of exposure was set to zero when $C_1< -1$ to induce partial non-overlap and reflect a one-sided violation of the positivity assumption. Thus, $\psi_{mRRT}$ represents an appropriate target parameter and $\psi_{mRR}$ is undefined. 
Scenario (b) removed this restriction so that $\psi_{mRR}$ is defined, whereas scenario (c) used a non-linear treatment model. 
In all scenarios, the analytic CC dataset was formed by randomly sampling  1,000 cases ($Y=1$) and 4,000 controls ($Y=0$).  Appendix G provides the full data-generating mechanisms and scenario (c) results.

\subsection*{Diagnostic checking and matching radius selection}

We present the results of the proposed diagnostic checks for the mRRT using a single simulated dataset. Table \ref{tab:bc_cc} presents the covariate balance before and after applying IPTW or matching. Matching radii were defined as $\rho_a=d_a\times \sigma_a$ for stage $a=1,2$, where $\sigma_a$ is the standard deviation of the logit of the propensity scores in the relevant subset (see Figure~\ref{ps matching algorithm}). Four radius values were considered:  $d_1 = \{0.05, 0.03, 0.01, 0.005\}$ and $d_2 = \{0.05, 0.03, 0.02, 0.01\}$.
In stage one, the treated and untreated controls initially exhibited substantial imbalance, with SMDs exceeding 0.7. Applying IPTW reduced the SMDs to below 0.01. For matching, covariate balance improved markedly as $d_1$ decreased, with SMD also falling below 0.01. A similar pattern was observed in stage two as $d_2$ decreased. We therefore selected $d_1=d_2=0.01$, which achieved a good balance. Figure \ref{fig: cc_ps} illustrates the pre- and post-weighting propensity score overlap. The distribution of treated and untreated controls initially showed non-overlap of untreated with propensity scores near zero, but this is not a violation when estimating the mRRT.  After weighting or matching, the histograms overlapped almost perfectly. In stage two (bottom row), pre-matching overlap between untreated cases and controls suggested that the propensity score model did not substantially extrapolate, and matching produced excellent alignment.

\begin{figure}[htbp]
 \centering
 \includegraphics[width=18cm,height=16cm]{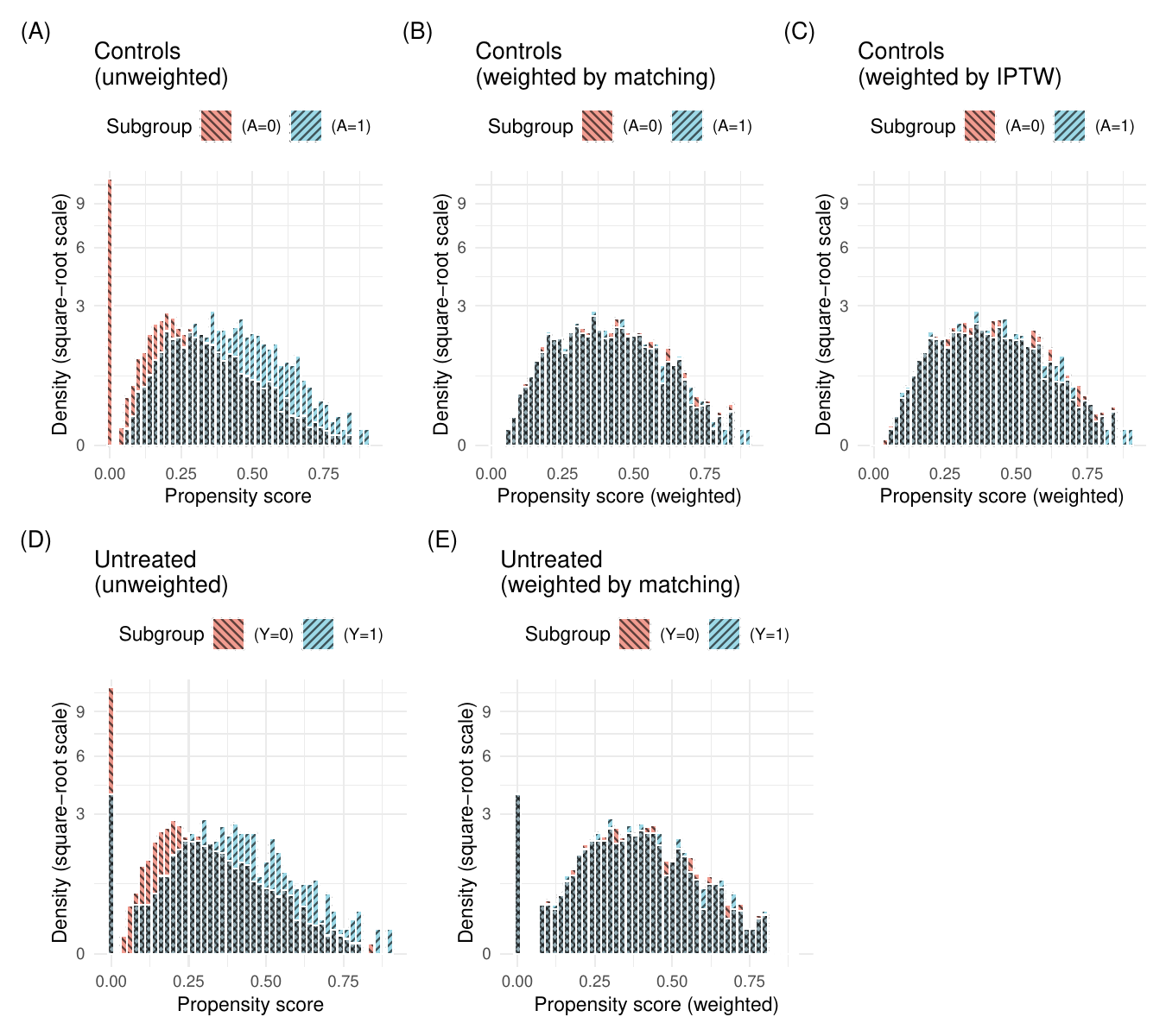}
 \caption{Diagnostic checks of the propensity score overlap for mRRT weighting  in a simulation (a) for a case-control study: comparison between the treated controls and untreated controls before weighting (A) and after weighting for matching with $d_1=d_2=0.01$ (B) and IPTW (C), and between untreated controls and untreated cases before (D) and after matching (E). 
 }
 \label{fig: cc_ps}
 \end{figure}

\begin{table}[ht]
\centering\small\captionsetup{width=.7\textwidth}
\caption{Diagnostic checks of covariate balance for the mRRT  in  case-control simulation (a): comparison between the treated controls and untreated controls before and after weighting for IPTW and matching, and between untreated cases and untreated controls before and after matching across matching radii. Variables are summarized as mean and standard deviation (mean (SD)); SMD: standardized mean difference. }\label{tab:bc_cc}\vspace{2mm}
\begin{tabular}{lllrrc}
  \toprule
   \multicolumn{6}{l}{ \emph{First stage: balance between treated and untreated controls}} \vspace{2mm} \\ 
& $d_1$ & Variable & $(A=0, Y=0)$ & $(A=1, Y=0)$ & SMD \\ 
 \cmidrule(lr){2-6}
\multirow{3}{*}{Before weighting}&  & Size &  2885 & 1115 & \\ 
&   &  $C_1$  & -0.210 (0.994)  & 0.476 (0.785)  &  0.767\\ 
&  &  $C_2$   &  -0.193 (0.974)  & 0.520 (0.950)  &  0.741\\ \cmidrule(lr){2-6}
\multirow{2}{*}{ IPTW} & \multirow{2}{*}{ }  &  $C_1$   & 0.470 (0.807)  &0.476 (0.784)  &0.007\\ 
&  & $C_2$  & 0.519 (0.960)  &0.520 (0.950) & 0.002\\ \cmidrule(lr){2-6}
\multirow{12}{*}{Matching}
 &  \multirow{3}{*}{0.05} & Size & 2252 & 1115 &  \\ 
& &  $C_1$   & 0.449 (0.795) & 0.476 (0.784) &   0.034 \\ 
&  & $C_2$    & 0.483 (0.944) & 0.520 (0.950) &  0.039 \\  \cmidrule(lr){2-6}
 & \multirow{3}{*}{0.03} & Size & 2247 & 1115 &  \\ 
&  &  $C_1$  & 0.462 (0.799) & 0.476 (0.784) &  0.018\\ 
&  & $C_2$   & 0.508 (0.959) & 0.520 (0.950) &  0.013 \\  \cmidrule(lr){2-6}
 & \multirow{3}{*}{0.01} & Size & 2225 & 1115 & \\ 
& &  $C_1$   &  0.465 (0.801) & 0.476 (0.784) &  0.014 \\ 
& & $C_2$   &  0.523 (0.966) & 0.520 (0.950) &  0.003\\   \cmidrule(lr){2-6}
 & \multirow{3}{*}{0.005}  & Size & 2210 & 1115 & \\ 
& &  $C_1$  &   0.464 (0.799) & 0.476 (0.784) &  0.015 \\
 &  &$C_2$  &  0.525 (0.969) & 0.520 (0.950) &  0.005\\ 
   \midrule
  \multicolumn{6}{l}{\emph{Second stage: balance of donor matching (independent of $d_1$) between untreated cases and controls}}\vspace{2mm}\\ 
& $d_2$ & Variable &  $(A=0, Y=0)$ & $(A=0, Y=1)$ & SMD\\ 
 \cmidrule(lr){2-6}
\multirow{3}{*}{Before weighting}&  & Size &  2885 &  536 & \\ 
&   &  $C_1$   & -0.210 (0.994) &  0.324 (0.954) &  0.549\\ 
&  &  $C_2$  &  -0.193 (0.974) & 0.416 (0.897) &  0.651\\ \cmidrule(lr){2-6}
\multirow{12}{*}{Matching}
  & \multirow{3}{*}{0.05} & Size & 2841 & 535 &  \\ 
  &  &  $C_1$  & 0.260 (0.903) &  0.321 (0.951)  & 0.065\\ 
 &   & $C_2$  & 0.391 (0.894)  &  0.411 (0.887)  & 0.022 \\  \cmidrule(lr){2-6}
   & \multirow{3}{*}{0.03}& Size & 2819 & 535 &  \\ 
 &  &  $C_1$  & 0.280 (0.917)  &  0.321 (0.951)  &  0.043 \\ 
 &  & $C_2$  & 0.412 (0.896)  & 0.411 (0.887)  &  0.002 \\  \cmidrule(lr){2-6}
  & \multirow{3}{*}{0.02} & Size & 2793 & 534 &  \\ 
 &  &  $C_1$   &  0.286 (0.922) &  0.318 (0.949) &  0.034\\ 
 &  & $C_2$   & 0.414 (0.888) & 0.405 (0.880) &   0.010\\  \cmidrule(lr){2-6}
   & \multirow{3}{*}{0.01}& Size & 2685 & 534 &  \\ 
 &   & $C_1$  & 0.289 (0.923) & 0.318 (0.949) &  0.030\\ 
 &  & $C_2$   & 0.419 (0.890) & 0.405 (0.880) &  0.015\\ 
 \bottomrule
\end{tabular}
\end{table}

\subsection*{Results}

Table \ref{tab:est_CC} summarizes the performance of the IPTW, one-step, and matching estimators for estimating the mRRT and mRR on 500 replications. In scenario (a), when targeting the mRRT (true value 0.93), the IPTW performed well with negligible bias and close agreement between MC and BS standard errors, yielding a coverage rate of 95.6\%. The matching estimator was sensitive to the radius: larger values, especially $d_1=d_2=0.05$, produced noticeable bias and lower coverage, reaching 82\%, whereas smaller radii improved both. 
For $d_1 \in (0.005, 0.01)$ and $d_2 \in (0.01, 0.02, 0.03)$, the bias was below 0.02 and the coverage rates remained high, typically between 95.8\% and 98\%. 

Scenario (b) evaluated the estimation of the mRR (true value 0.77). IPTW and one-step estimators performed well, producing estimates close to the true value and coverage rates near the nominal level (93.6\% and 94.6\%, respectively). In contrast, the logistic regression estimator showed higher bias (0.063) and lower coverage (80.6\%). In scenario (c), incorporating machine learning substantially improved performance relative to the parametric regression models, reducing bias and increasing coverage (Table S7 in Appendix G.4). This highlighted the benefits of flexible nuisance function estimation, though it may introduce more variance. An empirical demonstration of the double robustness of the one-step estimator is presented in Table S9 in Appendix G.4.

\begin{table}[ht]
\centering\small
\captionsetup{width=.60\textwidth}
\caption{Estimation summaries of IPTW, one-step (OS), and propensity score matching (with four different two-stage matching radius multipliers) for both mRRT and mRR in case-control simulations (a) and (b). Bias, Monte Carlo standard errors (MC SE), mean bootstrap-estimated standard errors (BS SE), and coverage rates.}\label{tab:est_CC}\vspace{2mm}
\begin{tabular}{lrrrrrrc}
  \toprule
Method & $d_1$ & $d_2$ & $\hat{\psi}$ & Bias & MC SE & BS SE & Coverage Rate (\%)\\  \midrule
  \multicolumn{8}{l}{\emph{Simulation (a): Estimation of mRRT  (true value: 0.933)}}\vspace{3mm}\\ 
IPTW  &  &   & 0.933 & 0.000 & 0.094 & 0.092 & 95.60 \\  \cmidrule(lr){2-8}
\multirow{16}{*}{Matching} & \multirow{4}{*}{0.05} & 0.05 & 1.031 & 0.098 & 0.102 & 0.100 & 82.00 \\ 
&  & 0.03 & 1.004 & 0.071 & 0.102 & 0.101 & 90.60 \\ 
&  & 0.02 & 1.002 & 0.069 & 0.100 & 0.101 & 90.40 \\ 
&  & 0.01 & 1.013 & 0.079 & 0.099 & 0.101 & 87.60 \\\cmidrule(lr){2-8}
& \multirow{4}{*}{0.03}  & 0.05 & 1.002 & 0.069 & 0.103 & 0.102 & 90.20 \\
&  & 0.03 & 0.971 & 0.038 & 0.103 & 0.104 & 94.20 \\ 
&   & 0.02 & 0.967 & 0.034 & 0.101 & 0.104 & 95.80 \\ 
&   & 0.01 & 0.979 & 0.046 & 0.099 & 0.104 & 95.00 \\ \cmidrule(lr){2-8}
& \multirow{4}{*}{0.01}  & 0.05 & 0.987 & 0.054 & 0.103 & 0.103 & 93.00 \\ 
&  & 0.03 & 0.952 & 0.019 & 0.104 & 0.106 & 96.20 \\ 
&  & 0.02 & 0.943 & 0.010 & 0.102 & 0.108 & 96.00 \\ 
&  & 0.01 & 0.944 & 0.011 & 0.100 & 0.109 & 98.00 \\  \cmidrule(lr){2-8}
& \multirow{4}{*}{0.005}  & 0.05 & 0.986 & 0.053 & 0.103 & 0.103 & 93.40 \\ 
&  & 0.03 & 0.950 & 0.017 & 0.104 & 0.106 & 96.20 \\ 
& & 0.02 & 0.940 & 0.007 & 0.102 & 0.108 & 95.80 \\ 
& & 0.01 & 0.937 & 0.004 & 0.100 & 0.110 & 97.60 \\   \midrule
  \multicolumn{8}{l}{\emph{Simulation (b): Estimation of conditional RR/mRR (true value: 0.770)}}\vspace{2mm}\\ 
  Logistic &  &   & 0.708  & 0.063 & 0.060 & 0.059 & 80.60 \\  
IPTW  &  &   & 0.777  & 0.006 & 0.066 & 0.066 & 93.60 \\  
OS  &  &   &  0.786 & 0.015 & 0.073 & 0.078 & 94.60  \\
   \bottomrule
\end{tabular}
\end{table}

\section*{Application}
To demonstrate the application of this methodology using real CC study data, we utilized data from the PRevention of OVarian cancer in Quebec (PROVAQ) study to estimate the effect of regular aspirin use on ovarian cancer risk. PROVAQ is a population-based CC study conducted in Montreal, Canada from 2011 to 2016.\cite{koushik2017hormonal} This study recruited female Canadian citizens aged 18--79 years who spoke French or English and resided in the Greater Montreal area. A total of 498 incident ovarian cancer cases (borderline and invasive) were recruited from seven hospitals (78\% participation rate), and 908 controls were randomly selected from the Quebec List of Electors (56\% participation). Controls were frequency-matched to cases by 5-year age strata and region. In examining a short survey given to non-participants, it was noted that participation among both cases and controls was associated with age and educational attainment. Participants reported information on sociodemographic characteristics, lifestyle, reproductive history, and medical history, including lifetime aspirin use.\cite{sarr_nonopioid_analgesic} Regular aspirin use was defined as use of at least one tablet per week for at least six consecutive months.
Overall, 15.3\% of participants reported regular aspirin use (13.7\% among cases and 16.2\% among controls). Descriptive characteristics of the study population by case status are presented in Table S10 in Appendix H.1. 

We implemented IPTW and propensity score matching for the estimation of the mRRT, while logistic regression, IPTW and the one-step doubly robust method were used for the mRR. The control-fitted propensity score model included age, education, the interaction between them, body mass index, history of endometriosis, history of menopausal hormone therapy, duration of oral contraceptive use, pack years of smoking, alcohol consumption, and regular use of other non-opioid analgesics. 
For matching, the radii were defined as 0.2 times the standard deviation of the logit-transformed propensity scores among controls in stage one and among non-regular aspirin users in stage two after the evaluation of different radius multipliers (see Table S12 in Appendix H). Under this specification, all cases who did not regularly use aspirin were successfully matched during the second stage, and among the options explored, this provided the best overall balance. Thus, the post-matching sample size was equal to the original. To reduce potential model misspecification, we also estimated the nuisance functions (propensity score and outcome probabilities) for IPTW and the one-step estimator using Super Learner (SL),
\cite{van2007SL} with random forest (\texttt{SL.randomForest}), generalized linear models (\texttt{SL.glm}), generalized additive models (\texttt{SL.gam}), and LASSO (\texttt{SL.glmnet}). When using SL, the estimated treatment and outcome probabilities were truncated at [0.001, 0.999] and [0.05, 0.95]. 

Diagnostic results of the propensity score overlap are presented in Figures S5-S7 in Appendix H.2. After weighting or matching, the distributions of propensity scores of treated and untreated controls were harmonized. The covariate SMDs (Table S11) were generally reduced after adjustment, with matching achieving overall balance for the mRRT. For the mRR, IPTW with the parametric method also improved balance, whereas SL-based IPTW led to less consistent balance. Second stage matching between untreated cases and untreated controls did not meaningfully improve balance, as the post-matching SMD seemed to be similar to the unweighted values and remained unbalanced for several covariates. 

Table \ref{tab:est_app} presents the estimated mRRT and mRR. All mRRT estimates suggested effect sizes of 0.75-0.78 with matching and IPTW with SL producing bootstrap confidence intervals that exclude the null. For the mRR, point estimates varied between methods relying on parametric models vs. SL. Logistic regression, IPTW, and the one-step method using generalized linear models gave comparable point estimates around 0.77-0.78, though IPTW and one-step had higher standard errors. However, when SL was used in IPTW and the one-step method, the point estimates decreased. For the one-step estimator, the usage of SL led to greatly increased standard errors. This was due to values of the propensity score $\pi_{aN}$ and of $1-\mu_{aN}$ close to zero, leading to large weights. Truncation of these probabilities controlled the standard error and resulted in point estimates of roughly 0.71-0.72 for both IPTW and one-step estimators with SL, though only the IPTW confidence intervals excluded the null.

\begin{table}[]
   \centering
\captionsetup{width=0.42\textwidth, font=footnotesize, skip=2pt}
\caption{Estimated mRR and mRRT in the PROVAQ study.  BS SE: Mean bootstrap-estimated standard error.}\label{tab:est_app}
\resizebox{1\width}{!}{%
\begin{tabular}{lrrc}
  \toprule
Method & Estimate & BS SE & 95\% CI\\  \midrule
  \multicolumn{4}{l}{\emph{Analysis (a): Estimation of mRRT  }}\vspace{1mm}\\ 
  IPTW  & 0.753 & 0.134 & [0.474, 1.007] \\  
  IPTW\_SL$_{0.001}$  & 0.777 & 0.125 & [0.469, 0.942] \\ 
  Matching &  0.758 & 0.137 & [0.460, 0.990]\\ \midrule
  \multicolumn{4}{l}{\emph{Analysis (b): Estimation of conditional RR/mRR }}\vspace{1mm}\\ 
  Logistic & 0.765 & 0.172 & [0.544, 1.067] \\
  IPTW  & 0.771 & 0.289 & [0.550, 1.627] \\
  IPTW\_SL$_{0.001}$ & 0.716 & 0.125 & [0.467, 0.942] \\
  IPTW\_SL$_{0.05}$ & 0.714 & 0.116 & [0.460, 0.897] \\
OS  &  0.781 & 0.315 & [0.506, 1.634]  \\
OS\_SL$_{0.001}$ & 0.674 & 0.617 & [0.082, 2.782] \\
OS\_SL$_{0.05}$ & 0.718 & 0.202 & [0.597, 1.404] \\
   \bottomrule
\end{tabular}}
\captionsetup{font=footnotesize}
\caption*{\textit{Notes:} IPTW\_SL and OS\_SL denote estimators in which nuisance functions were estimated using machine learning via Super Learner (SL.glm, SL.gam, SL.glmnet, SL.randomForest), where the subscripts ${0.001}$ and  ${0.05}$ indicate the truncation intervals [0.001, 0.999] and [0.05, 0.95], respectively. }
\end{table}

\section*{Discussion}

Logistic regression of the outcome on covariates and exposure is the standard approach in CC analysis. The rare outcome assumption that we make also implies that the usual adjusted logistic regression of the outcome approximately estimates risk ratios. However, logistic regression can only estimate populational effects when omitting interaction terms between exposure and any covariate; this omission may result in model misspecification, potentially biasing the resulting estimation.\cite{schnitzer2022estimands} Propensity score methods may be beneficial in that they are agnostic to outcome model specification and thus to treatment effect heterogeneity. While the doubly robust estimator requires estimation of the conditional expectation of the outcome in the sample, the double robustness property also allows for nice asymptotic properties while using machine learning, which avoids bias due to model misspecification.\cite{jiangTNDDR}

One limitation of the doubly robust estimator was increased variance in the presence of estimated conditional sample probabilities of the outcome close to one ($\mu_{aN}(\boldsymbol C) \approx 1$). This did not impact the IPTW estimator, which had lower variance both in the simulation study and application.  
We note that this variance inflation issue due to $\mu_{aN}(\boldsymbol C) \approx 1$ is specific to the CC sampling framework and is not a limitation of doubly robust methods in prospective designs. 

We broadly recommend applying our propensity score overlap and covariate balance metrics when conducting CC analyses with propensity scores. Note that these differ from the standard diagnostics used for prospective designs. 

Our application dataset used incidence density sampling, where controls are sampled among individuals at risk at the time a case occurs, and sampled conditional on their covariates in order to replicate the distribution observed in cases. While we did not explicitly investigate incidence density sampling theoretically or in the simulation study, we note that it satisfies the requirement that controls are selected independently of treatment conditional on covariates. Thus, we believe that our methods apply directly to this context.

Past work has demonstrated that CC designs that match controls to cases can reduce estimation variance.~\cite{rose2009match} 
Our convergence proof for the proposed two-step propensity score matching estimator suggests that the matching estimator implicitly constructs weights approximating treatment-odds weighting within propensity score neighborhoods, while the matching radius serves as a tuning parameter that may improve finite-sample stability. 
Future work could explore alternative balancing-weight implementations for the mRRT, develop doubly robust and targeted estimators for the mRRT, integrate matched CC designs into our framework and extend the methodology to multilevel categorical exposures.

\vskip 8pt
\section*{Conflicts of interest} The authors declare no potential conflict of interests.
\vskip 8pt
\section*{Funding} This work was supported by an NSERC Discovery Grant. AK holds the McGill University Chair in Community Cancer Care at St. Mary's. DT is supported by a research career award from the Fonds de recherche du Québec -- Santé (\url{https://doi.org/10.69777/379793}, \url{https://doi.org/10.69777/312198}). MES holds a tier 2 Canada Research Chair in Machine Learning and Causal Inference. 
\vskip 8pt
\section*{Data availability statement} R code used to generate the simulated data and conduct the simulation studies will be made available on GitHub (to be uploaded) upon publication. The data used in the applied analysis are not publicly available because of privacy, ethical, and data-use restrictions. Access to these data may be requested from anita.koushik@mcgill.ca, and can be shared contingent on reasonable request and institutional approval.

\clearpage
\bibliographystyle{ama}
\bibliography{references}
\clearpage

%%%%%%%%%%%%%%%%%%%%%%%%%%%%%%%%%%%%%%%%%%%%%%%%%%%%%
%%%%%%%%%%%%%%%%%%%%%%%%%%%%%%%%%%%%%%%%%%%%%%%%%%%%%

\begin{center}
    \textbf{\Large{Appendix for ``Causal Inference via Propensity Scores for Case-control Studies''}}
    \\\vspace{0.8cm}
    by Yan Liu, Anita Koushik, Philippe Boileau,
Cong Jiang, Miceline M\'esidor, Claudia Waddingham,\\
Denis Talbot, and Mireille E. Schnitzer
\end{center}
\appendix
\section{Proof of inverse probability weighted identifiability for effects in the complete population}
For any of the three case-control (CC) designs, we assume that if $Y=1$ then $S=1$, that is, the presence of the outcome implies that the individual is eligible for study inclusion. We also assume  that $\mathbb{P}_{CC}(A=a\mid Y=0, \boldsymbol{C})= \mathbb{P}(A=a\mid \boldsymbol{C})$. Then we can identify the following up to the constant $q=\mathbb{P}(S=1)$:
\begin{align*}
&\mathbb{E}_{CC}\left\{Y \mathbb{I}(A=a) q \middle/ \mathbb{P}(A=a\mid \boldsymbol{C})\right\}, \\
&=\mathbb{E}_{CC}\left\{Y \mathbb{I}(A=a)  \middle/ \mathbb{P}(A=a\mid \boldsymbol{C})\right\}\mathbb{P}(S=1), \\
&=\mathbb{E}\left\{Y \mathbb{I}(A=a)  \middle/ \mathbb{P}(A=a\mid \boldsymbol{C})\mid S=1\right\}\mathbb{P}(S=1), \\
&=\mathbb{E}\left\{YS\mathbb{I}(A=a) \middle/ \mathbb{P}(A=a\mid \boldsymbol{C})\right\},\\
&=\mathbb{E}\left\{Y\mathbb{I}(A=a) \middle/ \mathbb{P}(A=a\mid \boldsymbol{C})\right\}, \text{ which is the usual IPTW g-formula for the treatment-specific mean outcome;}\\
&=\mathbb{E}\{\mathbb{P}(Y=1\mid A=a, \boldsymbol{C}) \}.
\end{align*}

\section{Proof of inverse probability weighted identifiability for effects in the treated population}
This proof is identical to the previous except involves constant $q_1=\mathbb{P}(S=1\mid A=1)$. The following is identifiable up to $q_1$ under the assumption that the propensity scores can be estimated using control data.
\begin{align*}
 &\mathbb{E}_{CC}\left\{Y (1-A) q_1 \mathbb{P}(A=1\mid \boldsymbol{C})\middle/ \mathbb{P}(A=0\mid \boldsymbol{C}) \mid A=1\right\},\\
 &=\mathbb{E}_{CC}\left\{Y (1-A)  \mathbb{P}(A=1\mid \boldsymbol{C})\middle/ \mathbb{P}(A=0\mid \boldsymbol{C}) \mid A=1\right\}\mathbb{P}(S=1\mid A=1),\\
 &=\mathbb{E}\left\{Y (1-A)  \mathbb{P}(A=1\mid \boldsymbol{C})\middle/ \mathbb{P}(A=0\mid \boldsymbol{C}) \mid A=1, S=1\right\}\mathbb{P}(S=1\mid A=1),\\
 &=\mathbb{E}\left\{Y S (1-A)  \mathbb{P}(A=1\mid \boldsymbol{C})\middle/ \mathbb{P}(A=0\mid \boldsymbol{C}) \mid A=1\right\},\\
 &=\mathbb{E}\left\{Y (1-A)  \mathbb{P}(A=1\mid \boldsymbol{C})\middle/ \mathbb{P}(A=0\mid \boldsymbol{C}) \mid A=1\right\}, \text{ the usual IPTW formula for the average treatment effect in the treated;}\\
 &=\mathbb{E}\{\mathbb{P}(Y=1\mid A=0, \boldsymbol{C})\mid A=1 \}. 
\end{align*}

\section{Test-negative design}
The TND specifically aims to estimate vaccine effectiveness against an infectious disease. Such a study recruits patients seeking care for  symptoms that resemble the disease of interest (target disease) who then undergo testing to confirm whether they are infected with the target disease. To define a control group, the TND requires that symptoms resembling those of the target disease may also arise due to some other pathogen. We denote $W_1$ as the presence of symptoms due to the target pathogen and $W_0$ as symptoms due to another pathogen. Neither $W_1$ nor $W_0$ are known at the time of recruitment. Thus, we use $W$ to indicate the presence of target-disease-like symptoms  due to either infection, i.e. $W = \max\{W_1, W_0\}$, which is observed. We define $S$ as care-seeking for symptoms (so that $S=1$ also implies the presence of symptoms $W=1$) which represents the primary inclusion criteria of the TND. 
The outcome of interest, care-seeking for the target disease, is given as $Y = \mathbb{I}(W_1=1, S=1)$.  We again define $q = \mathbb{P}(S=1)$ to be the probability or prevalence of the inclusion criterion. Thus, while the complete data are defined as $(\boldsymbol{C},A,W_1,W_0,S)$, we observe only $N$ i.i.d. draws from $(\boldsymbol{C},A, Y)\mathbb{I}(S=1)$, that is, only confounders, vaccination status, and case ($Y=1$) or control ($Y=0$) status in the TND sample.

 The absence of co-infections ($W_1=1$ and $W_0=1$ simultaneously) in addition to our DAG implies control exchangeability.\cite{jiangTNDDR} 
We have shown previously~\cite{schnitzer2022estimands} that under this DAG, control exchangeability, and positivity s.t., $0<\mathbb{P}(A=a\mid \boldsymbol{C})<1$ almost surely, the marginal risk ratio $\psi_{mRR}$ 
can be identified through the same IPTW formula as in Equation (2) in manuscript.

\section{Proof of convergence of the matching estimator}
This proof applies to a CC where the ratio of controls to cases is held constant as sample size increases or to a test-negative design (TND) (where this ratio is not controlled by the design). We assume control exchangeability or rare outcome and  a weaker positivity condition that $P(A=1\mid \boldsymbol{C})<1$ almost surely. To arrive at a causal parameter which is the analogue of the mRR, we additionally assume consistency ($A=a\Rightarrow Y=Y^{(a)}$ for $a=(0,1)$) and one-sided conditional exchangeability $Y^{(0)} \perp A \mid \boldsymbol{C}$.

\begin{theorem}\label{matchingconv}
As $N\rightarrow \infty$, $\hat{\psi}_{mRRT} \rightarrow_{\mathbb{P}_{TND}} \psi_{mRRT}$.
\end{theorem}
\begin{proof}
In the numerator of $\hat{\psi}_{mRRT}$, we have that 
\begin{align}
    \left.\sum_{i=1}^N a_i y_i \middle/ \sum_{i=1}^N a_i\right. \rightarrow_{\mathbb{P}_{CC}} \quad &\mathbb{E}(AY\mid S=1)/ \mathbb{P}(A=1\mid S=1),\notag\\
    =&\mathbb{E}(Y\mid A=1, S=1),\notag\\
    =&\frac{\mathbb{P}(Y=1,S=1\mid A=1)}{\mathbb{P}(S=1\mid A=1)},\notag\\
    =&\frac{\mathbb{P}(Y=1\mid A=1)}{\mathbb{P}(S=1\mid A=1)}, \text{ because } Y=1 \text{ implies } S=1;\notag\\
    =&\mathbb{P}(Y=1\mid A=1) / q_1,\notag\\
    =&\mathbb{P}(Y^{(1)}=1\mid A=1) / q_1\text{ since when }A=1\text{ we observe } Y^{(1)}.\notag
\end{align}

For the denominator, following Li and Greene~\cite{LiGreene_matchingweights}, we now suppose that the propensity score can only take on finite values $\{p_k, k=1,...,K\}$ where each $p_k<1$. We let the radius be $\rho_1=0$. Let $M_{k}$ be the number of times an arbitrary untreated control participant is matched to any treated control participant with propensity score $p_k$.

First, we show that the mean weight within each propensity score stratum converges to the target weight times the probability mass of the stratum. 
\begin{align}
&\mathbb{E}\{M_k\mid \boldsymbol{C}=\boldsymbol{c}\},\notag\\
&= \mathbb{E}\{M_k\mid \boldsymbol{C}=\boldsymbol{c},S=1, Y=0\},\text{ by outcome rarity [CC] or control exchangeability [TND];}\notag\\
   &=  \mathbb{E}\{ M_k\mid \pi_N(\boldsymbol{c})=p_k, \boldsymbol{C}=\boldsymbol{c}, S=1,Y=0\}\mathbb{P}\{ \pi_N(\boldsymbol{c})=p_k \mid \boldsymbol{C}=\boldsymbol{c}, S=1, Y=0\}, \notag\\
   &\text{by the law of total expectation and due to the definition of $M_k$ only taking non-zero values within the stratum $p_k$;}\notag\\
   &= \frac{\mathbb{P}\{ A=1\mid \pi_N(\boldsymbol{c})=p_k, \boldsymbol{C}=\boldsymbol{c}, S=1, Y=0\}}{\mathbb{P}\{ A=0\mid \pi_N(\boldsymbol{c})=p_k, \boldsymbol{C}=\boldsymbol{c}, S=1, Y=0\}}\mathbb{P}\{ \pi_N(\boldsymbol{c})=p_k \mid \boldsymbol{C}=\boldsymbol{c}, S=1, Y=0\},\notag\\
  &\text{since the expected weight is the probability that there exists a treated (control) participant in the stratum times the inverse density of}\notag\\
  &\text{untreated (control) participants in the stratum;}\notag\\
  &=  \frac{\mathbb{P}\{ A=1\mid \pi_N(\boldsymbol{c})=p_k, \boldsymbol{C}=\boldsymbol{c}\}}{\mathbb{P}\{ A=0\mid \pi_N(\boldsymbol{c})=p_k, \boldsymbol{C}=\boldsymbol{c}\}}\mathbb{P}\{ \pi_N(\boldsymbol{c})=p_k \mid \boldsymbol{C}=\boldsymbol{c}\}, \text{ by control exchangeability/rare outcome;}\notag\\
  &\rightarrow_N \frac{\mathbb{P}\{ A=1\mid  \boldsymbol{C}=\boldsymbol{c}\}}{\mathbb{P}\{ A=0\mid  \boldsymbol{C}=\boldsymbol{c}\}}\mathbb{P}\{ \pi(\boldsymbol{c})=p_k \mid \boldsymbol{C}=\boldsymbol{c}\}, \text{ assuming consistency of the propensity score estimation, by properties of the}\notag\\
  &\text{propensity score.}\notag
\end{align}
This means that for an untreated case observation $j$ that uses untreated control observation $i^*$'s weight, the large-sample mean of the weight will approximate $\pi(\boldsymbol{C})/\{1-\pi(\boldsymbol{C})\}$.

Now denote by $m_{j,k}=\{m_j \text{ if  }\pi_N(C_j)=p_k; 0\text{ otherwise}\}$ the weight of individual $j$ corresponding to the propensity score stratum $k$;  $M_k$ is the corresponding random variable.  Note that if we assume that the form of the propensity score is known, all of the randomness of $M_k$ comes from $\boldsymbol{C}$.
The matching estimator in the denominator converges to
\begin{align}
 &\left.\sum_{j=1}^N (1-a_j)y_j m_j\middle/ \sum_{j=1}^N a_j 
 \right.  \quad= \left.\sum_{j=1}^N \sum_{k=1}^K(1-a_j)y_j m_{j,k} \middle/ \sum_{j=1}^N a_j \right. \notag \\ 
 &\rightarrow_N%^{\mathbb{P}_{TND}}  
 \frac{\mathbb{E}\left\{\sum_{k=1}^K (1-A)YM_{k}\mid S=1\right\}}{\mathbb{P}(A=1\mid S=1)},\notag\\
  &=  
 \frac{\mathbb{E}\left\{\sum_{k=1}^K (1-A)YM_{k}\mid S=1\right\}\mathbb{P}(S=1)}{ \mathbb{P}(S=1\mid A=1)\mathbb{P}(A=1)},\notag\\
   &=  
 \frac{\sum_{k=1}^K\mathbb{E}\left\{ (1-A)YSM_{k}\right\}}{ \mathbb{P}(S=1\mid A=1)\mathbb{P}(A=1)},\notag\\
   &=  
 \frac{\sum_{k=1}^K \mathbb{E}\left\{(1-A)YM_{k}\right\}}{ \mathbb{P}(S=1\mid A=1)\mathbb{P}(A=1)},\text{ because the outcome }Y=1\text{ implies } S=1,\notag\\
  &=\frac{\sum_{k=1}^K \mathbb{E}\left\{(1-A)Y^{(0)}M_{k}\right\}}{ \mathbb{P}(S=1\mid A=1)\mathbb{P}(A=1)}, \text{because when }A=0\text{, we observe }Y^{(0)},\notag\\
 &= \frac{\sum_{k=1}^K\mathbb{E}\left[  \mathbb{E}\{Y^{(0)}(1-A)M_{k} \mid \boldsymbol{C}\}\right] }{\mathbb{P}(S=1\mid A=1)\mathbb{P}(A=1)},\notag\\
 &=\frac{\mathbb{E}\left[ \mathbb{E}\{Y^{(0)}\mid \boldsymbol{C}\} \sum_{k=1}^K\mathbb{E}\{(1-A)M_{k} \mid \boldsymbol{C}\}\right] }{\mathbb{P}(S=1\mid A=1)\mathbb{P}(A=1)},\text{ because $Y^{(0)}$ is assumed independent of treatment conditional on covariates;}\notag\\
  &=\frac{\mathbb{E}\left[ \mathbb{E}\{Y^{(0)}\mid \boldsymbol{C}\} \sum_{k=1}^K\mathbb{P}(A=0\mid M_{k}, \boldsymbol{C}) \mathbb{E}(M_{k}\mid \boldsymbol{C})\right] }{\mathbb{P}(S=1\mid A=1)\mathbb{P}(A=1)},\notag\\
 &=\frac{\mathbb{E}\left[ \mathbb{E}\{Y^{(0)}\mid \boldsymbol{C}\} \sum_{k=1}^K\mathbb{P}(A=0\mid  \boldsymbol{C})\frac{\mathbb{P}(A=1\mid \boldsymbol{C})} {\mathbb{P}(A=0\mid \boldsymbol{C})}\mathbb{P}\{\pi(\boldsymbol{C})=p_k\mid \boldsymbol{C}\}\right] }{\mathbb{P}(S=1\mid A=1)\mathbb{P}(A=1)},\text{ subbing in the previous result and by properties of the}\notag\\
 &\text{propensity score,}\notag\\
 &=\frac{\mathbb{E}\left[ \mathbb{E}\{Y^{(0)}\mid \boldsymbol{C}_j\} \sum_{k=1}^K\mathbb{P}(A=1\mid \boldsymbol{C}) \mathbb{P}\{\pi(\boldsymbol{C})=p_k\mid \boldsymbol{C}\}\right] }{\mathbb{P}(S=1\mid A=1)\mathbb{P}(A=1)},\notag\\
  &=\frac{\mathbb{E}\left[ \mathbb{E}\{Y^{(0)}\mid \boldsymbol{C}\} \mathbb{P}(A=1\mid \boldsymbol{C})\right] }{\mathbb{P}(S=1\mid A=1)\mathbb{P}(A=1)},\text{ because }\sum_{k=1}^K\mathbb{P}\{\pi(\boldsymbol{C})=p_k\mid \boldsymbol{C}\}=1, \notag\\
   &=\frac{\mathbb{P}\{Y^{(0)}=1,A=1 \}  }{\mathbb{P}(S=1\mid A=1)\mathbb{P}(A=1)}, \notag\\
      &=\frac{ \mathbb{P}\{Y^{(0)}=1\mid A=1\} }{q_1}. \notag
\end{align}

\end{proof}

\section{Propensity score matching pseudocode}
Table \ref{tab:ps matching algorithm} presents the two-stage propensity score matching pseudocode for the estimation of marginal risk ratio among the treated (mRRT).

\begin{table}
    \centering\small
    \caption{\textcolor{black}{Propensity score matching pseudocode for estimation of $\psi_{mRRT}$} \label{tab:ps matching algorithm} }\color{black}
   \begin{tabular}{lllll}
\toprule
\textbf{Step} & \multicolumn{3}{l}{\textbf{Description}} \\
\midrule
\raisebox{1.5ex}{1} & \multicolumn{4}{l}{{\makecell[l]{With data structure $O=\{(Y_i, \boldsymbol{C}_i, A_i), i = 1, \cdots, N\}$, fit a propensity score model for $\pi = Pr(A=1|\boldsymbol{C})$ using  only the\\ controls ($Y=0$), then make predictions for all participants to obtain $\hat{\pi}_i; i=1,...,N$. }} }\\\midrule 
2 & \multicolumn{4}{l}{{\makecell[l]{Split data into cases ($Y=1$) and controls ($Y=0$). }} }\\
\midrule \vspace{1mm}
3 & \multicolumn{4}{l}{First stage: matching for covariate balance among the controls,} \\ \vspace{1mm}
 & \raisebox{1.5ex}{3.1} & \multicolumn{3}{l}{\makecell[l]{Compute the empirical standard deviation of the logit propensity scores among controls,\\
\(\hat\sigma^{{ctrl}} = \mathrm{SD}\{\text{logit}(\hat\pi_i): Y_i=0\}\).
      Define the radius $\rho_1=d_1\times \hat{\sigma}^{ctrl}$ where \(d_1 > 0\) is a user–chosen constant.}} \\\vspace{1mm}
 & \raisebox{1.5ex}{3.2} & \multicolumn{3}{l}{\makecell[l]{Define sets of treated and untreated groups in controls, $S_{1}^{ctrl} = \{i: Y_i=0, A_i=1\},S_{0}^{ctrl}= \{i: Y_i=0, A_i=0\}$. \\ Build a distance matrix $D^{ctrl}$ with entries $|\text{logit}(\hat{\pi}_j)- \text{logit}(\hat{\pi}_\ell)|$ for $j\in S_{1}^{ctrl}$ and $\ell\in S_{0}^{ctrl}$.}} \\\vspace{1mm}

&  3.3 & \multicolumn{3}{l}{Given radius $\rho_1$, for each treated control $j$ for $j\in S_{1}^{ctrl}$:  } \\\vspace{1mm}

 & & \raisebox{1.5ex}{3.3.1}  & \multicolumn{2}{l}{\makecell[l]{Define the match set  $L_j = \{\ell\in S_0^{ctrl}: D^{ctrl}_{j,l} \leq \rho_1\}$. If there are no eligible matches, set $L_j=\ell^*$ such that \\$D^{ctrl}_{j,\ell^*} = \text{min}(D^{ctrl}_{j,\ell})$.}}\\  \vspace{1mm}
 & & 3.3.2 & \multicolumn{2}{l}{\makecell[l]{Define a matrix $m$ with dimension $|S_{1}^{ctrl}|\times |S_{0}^{ctrl}|$, with $m_{j,\ell} = \mathbbm{1}_{(\ell\in L_j)}/\lvert L_j\rvert$. }}\\
\vspace{1mm}

& 3.4 & \multicolumn{3}{l}{Compute the weight $m_\ell$ for each untreated control as, $m_\ell=\sum_{j\in S_1^{trl}} m_{j,\ell}$.}\\
\midrule \vspace{1mm}

4 & \multicolumn{4}{l}{Second stage: donor matching,} \\ \vspace{1mm}
& \raisebox{1.5ex}{4.1} & \multicolumn{3}{l}{\makecell[l]{Compute the empirical standard deviation of the logit propensity scores among the untreated group,\\
\(\hat\sigma^{{untreat}} = \mathrm{SD}\{\text{logit}(\hat\pi_i): A_i=0\}\).
Define the radius $\rho_2=d_2\times \hat{\sigma}^{untreat}$ where $d_2 > 0$ is a user–chosen constant.}} \\\vspace{1mm}
 & \raisebox{1.5ex}{4.2}  & \multicolumn{3}{l}{\makecell[l]{Define sets of treated and untreated groups in cases, $S_{1}^{case}= \{i: Y_i=1, A_i=1\}$ and $S_{0}^{case}= \{i: Y_i=1, A_i=0\}$. \\ Build a distance matrix $D^{untreat}$ with entries $|\text{logit}(\hat{\pi}_k)- \text{logit}(\hat{\pi}_\ell)|$ for $k\in S_{0}^{case}$ and $\ell\in S_{0}^{ctrl}$.}}\\\vspace{1mm}

 &  4.3 & \multicolumn{3}{l}{Given radius $\rho_2$, then for each untreated case $k$ for $k\in S_{0}^{case}$,  } \\   \vspace{1mm}
 & & 4.3.1 & \multicolumn{2}{l}{\makecell[l]{Define the match set  $L_k = \{\ell\in S_0^{ctrl}: D^{untreat}_{k,\ell} \leq \rho_2\}$. If there are no eligible matches, set $L_k=\varnothing$. }}\\  \vspace{1mm}
 & & 4.3.2 & \multicolumn{2}{l}{\makecell[l]{ Compute the weight for each untreated case $k$, $
m_k =
\begin{cases}
0, & L_k=\varnothing,\\[6pt]
\displaystyle \frac{1}{\lvert L_k\rvert}\sum_{\ell\in L_k} m_\ell, & \lvert L_k\rvert>0.
\end{cases}
$}}\\ \midrule  \vspace{1mm}

5 & \multicolumn{4}{l}{For each treated case $t\in S_1^{case} $, define $m_t=1$.} \\  \midrule \vspace{1mm}

\raisebox{1.5ex}{6} & \multicolumn{4}{l}{\makecell[l]{The mRRT is estimated as the ratio of the weighted mean outcomes among the treated and untreated groups, i.e., \\ $\hat{\psi}_{mRRT} =  \sum_{i=1}^N Y_iA_i / \sum_{i=1}^N Y_i(1-A_i)m_i$. }} \\ 
\bottomrule
\end{tabular}
\end{table}

\section{Covariate balance checking algorithms for IPTW and propensity score matching}

Table \ref{IPTW_bc_algorithm} gives the covariate balance checking algorithms for inverse probability of treatment weighting (IPTW) in the controls. Table \ref{matching_bc_algorithm} gives the covariate balance checking algorithms for matching in two stages.

\begin{table}
    \centering\small
    \caption{\textcolor{black}{Covariate balance checking algorithm for IPTW} \label{IPTW_bc_algorithm} }\color{black}
   \begin{tabular}{lllll}
\toprule
\textbf{Step} & \multicolumn{3}{l}{\textbf{Description}} \\
\midrule\vspace{1mm}
\raisebox{1.5ex}{1} & \multicolumn{4}{l}{{\makecell[l]{With data structure $O=\{(Y_i, \boldsymbol{C}_i, A_i), i = 1, \cdots, N\}$, fit a propensity score model for $\pi = Pr(A=1|\boldsymbol{C})$ using  only \\the controls ($Y=0$), then make predictions for all participants to obtain $\hat{\pi}_i; i=1,...,N$. }} }\\\midrule \vspace{1mm}
 \raisebox{4ex}{2} & \multicolumn{4}{l}{{\makecell[l]{\vspace{2mm}Compute the IPTW weights ($w_i^{IPTW}$): \\\vspace{1mm}For mRRT: $w_i =
\mathbb{I}(A_i=1) + \mathbb{I}(A_i=0) \frac{\hat{\pi}_i}{1-\hat{\pi}_i}$,\\For mRR: $w_i = \frac{\mathbb{I}(A_i=1)}{\hat{\pi}_i} + \frac{\mathbb{I}(A_i=0)}{1-\hat{\pi}_i}$.}} }\\\midrule
\vspace{1mm}
3 & \multicolumn{4}{l}{\makecell[l]{Among the controls, define sets of treated and untreated groups, $S_{1} = \{i: Y_i=0, A_i=1\},S_{0}= \{i: Y_i=0, A_i=0\}$
}} \\\vspace{1mm}
& \raisebox{2.6ex}{3.1} & \multicolumn{3}{l}{\makecell[l]{Compute weighted means and variance:\\For each covariate $C \in \boldsymbol{C}$, calculate the weighted mean ($\bar{C}_{g}$) and weighted variance ($\hat{\sigma}^2_{g}$) for each set $g \in \{1, 0\}$:\\
$\bar{C}_{g} = \{\sum_{i \in \mathcal{S}_g} w_i C_i\}/\{\sum_{i \in \mathcal{S}_g} w_i\}$, $\qquad \hat{\sigma}^2_{g} = \{\sum_{i \in \mathcal{S}_g} w_i (C_i - \bar{C}_{g})^2\}/\{\sum_{i \in \mathcal{S}_g} w_i\}$. 
}} \\ \vspace{1mm}
& \raisebox{1.5ex}{3.2} & \multicolumn{3}{l}{\makecell[l]{Compute Weighted Standardized Mean Differences (SMD):\\$SMD = \{|\bar{C}_{1} - \bar{C}_{0}|\}\big/\{\sqrt{(\hat{\sigma}^2_{1} + \hat{\sigma}^2_{0})/{2}}\}$.
}} \\
\bottomrule
\end{tabular}
\end{table}

\begin{table}
    \centering\small
    \caption{\textcolor{black}{Covariate balance checking algorithm for propensity score matching for the mRRT} \label{matching_bc_algorithm} }\color{black}
   \begin{tabular}{lllll}
\toprule
\textbf{Step} & \multicolumn{3}{l}{\textbf{Description}} \\
\midrule\vspace{1mm}
\raisebox{1.5ex}{1} & \multicolumn{4}{l}{{\makecell[l]{With data structure $O=\{(Y_i, \boldsymbol{C}_i, A_i), i = 1, \cdots, N\}$, fit a propensity score model for $\pi = Pr(A=1|\boldsymbol{C})$ using  only \\the controls ($Y=0$), then make predictions for all participants to obtain $\hat{\pi}_i; i=1,...,N$. }} }\\\midrule \vspace{1mm}
2 & \multicolumn{4}{l}{Compute the matching weights  ($w_i^{Matching}$) for two stages: } \\ \vspace{1mm}
 & \raisebox{6.5ex}{2.1} & \multicolumn{3}{l}{\makecell[l]{First stage: define sets of treated and untreated groups among the controls, \\$S_{1}^{ctrl} = \{i: Y_i=0, A_i=1\},\quad S_{0}^{ctrl}= \{i: Y_i=0, A_i=0\}$.\vspace{1mm}\\ Based on the step 3 in the propensity score matching algorithm (Table \ref{tab:ps matching algorithm}), the weights in first stage are \\$w_i^{Matching_1} =
\begin{cases}
1, & \text{if}~i\in S_1^{ctrl},\\[6pt]
m_i, & \text{if}~i\in S_0^{ctrl}.
\end{cases}$}} \\ \cmidrule(lr){2-3}
\vspace{1mm}
& \raisebox{5,5ex}{2.2} & \multicolumn{3}{l}{\makecell[l]{Second stage: define the set of untreated cases, $S_{0}^{case} = \{i: Y_i=1, A_i=0\}$.\vspace{1mm}\\ Based on the match set $L_k$ where $k\in S_0^{case}$ defined at step 4.31 in Table \ref{tab:ps matching algorithm}, the weights in second stage are \\ $w_i^{Matching_2} =
\begin{cases}
1, & \text{if}~i\in S_0^{case},\\[6pt]
\sum_{(k:\lvert L_k\rvert>0)} \frac{\mathbbm{1}(i\in L_k)}{\lvert L_k \rvert}, & \text{if}~i\in S_0^{ctrl}.
\end{cases}$ }} \\\midrule \vspace{1mm}
3 & \multicolumn{4}{l}{\makecell[l]{Balance checking for matching in two stages based on the steps 3.1 and 3.2 in Figure \ref{IPTW_bc_algorithm}: } } \\ \vspace{1mm}
&3.1 & \multicolumn{3}{l}{\makecell[l]{For first stage: compute the weighted mean, variance and SMD between sets $S_1^{ctrl}, S_0^{ctrl}$. } } \\ 
& 3.2 & \multicolumn{3}{l}{\makecell[l]{For second stage: compute the weighted mean, variance and SMD between sets $S_0^{case}, S_0^{ctrl}$. } } \\ 
\bottomrule
\end{tabular}
\end{table}

\section{Simulation studies}

\subsection{Data-generating mechanisms for CC and TND studies}

Table \ref{tab:dg-cc} presents the data-generating mechanism for simulation scenarios (a), (b) and (c) of the CC study. The treatment assignment mechanisms for mRRT and mRR were varied to illustrate settings where each parameter would be most of interest; specifically, the mRRT is more interesting to estimate when some untreated individual have zero probability of being treated (so the mRR does not exist) but all treated individuals had a non-zero probability of being untreated. To form each dataset, 1,000 cases and 4,000 controls were randomly sampled. 

\begin{table}[H]\centering\small\captionsetup{width=.68\textwidth}
\caption{\small Data-generating mechanism of simulation scenarios (a), (b), and (c) for the CC study. Treatment assignment: A(a) for scenario (a), A(b) for scenario (b), and A(c) for scenario (c); Outcome: Y(a,b) for scenarios (a) and (b), and Y(c) for scenario (c).}\vspace{1mm}\label{tab:dg-cc}
\begin{tabular}{ll}
\toprule
\textbf{Variable} & \textbf{Generating Mechanism} \\
\midrule
$C_j,\ j = (1,2)$ & $\sim \text{Multivariate normal distribution as } \mathcal{N}\!\left(
\begin{bmatrix} 0 \\ 0 \end{bmatrix},
\begin{bmatrix}
1 & 0.3\\
0.3 & 1
\end{bmatrix}\right)$ \\[8pt]
$A (a) $ & $\sim\mathrm{Bernoulli}(p_V)$ where $p_V = \left\{
\begin{aligned}
& 0 \quad \text{if } C_1 < -1,\\
& \text{logit}^{-1}(-1 + 0.4\,C_1 + 0.7\,C_2)  \quad \text{otherwise}\\
\end{aligned}
\right.$\\[3pt]
$A (b)$ & $\sim\mathrm{Bernoulli}(p_V)$ where $p_V = \text{logit}^{-1}(-1 + 0.4\,C_1 + 0.7\,C_2)$\\[3pt]
$A (c)$ & $\sim\mathrm{Bernoulli}(p_V)$ where $p_V = \text{logit}^{-1}(-1 + 0.4\,C_1 + 0.7\,C_2+ 0.4\,C_1\,C_2)$\\[3pt]
$Y (a,b)$ & $\sim \mathrm{Bernoulli}\bigl(\text{logit}(p)=-4.95+ 0.4C_1 + 0.7C_2 + 0.7AC_1- 0.9A\bigr)$ \\
$Y (c)$ &  $\sim \mathrm{Bernoulli}\bigl(\text{logit}(p)=-5.07 + 0.4C_1 + 0.7C_2 + 0.7AC_1-0.9A\bigr)$  \\[3pt]\bottomrule
\end{tabular}
\end{table}

We also generated data from a TND. We first generated population data of sample size N= $2\times 10^6$. Each individual was assigned two baseline covariates, continuous $C_1$ and binary $C_2$. Vaccination status $A$ was then generated based on the two measured covariates. For the setting in which $\psi_{mRRT}$ is well defined, partial non-overlap was induced by setting $p_A=0$ for $C_1<-1$. We also considered a setting without this restriction, in which $\psi_{mRR}$ is defined. Individuals could become infected with a non-target pathogen ($I_0$), or the target pathogen ($I_1$). Symptom indicators ($W_0, W_1$) were generated conditional on corresponding infection status but were unobserved in the data. The composite indicator $W=max(W_0, W_1)$ representing any symptomatic presentation was observed. Hospitalization ($H$) was then simulated conditional on $W$. The TND sample consisted of 5,000 individuals randomly selected among those hospitalized ($H=1$). The observed TND data therefore has the structure $(C_1, C_2, A, Y)$ where  $Y=I_1$ indicates cases status for the target pathogen. Table \ref{tab:dg_TND} presents the data-generating mechanism for the  TND simulation study.

\begin{table}[H] 
\centering\small\captionsetup{width=.8\textwidth}
\caption{\small Data generating mechanism for the TND simulation. Treatment assignment: mRRT [A(a)] and mRR [A(b)]}\vspace{2mm}
\label{tab:dg_TND}
 \resizebox{1\width}{!}{%
 \begin{tabular}{lll}
 \\[-1.5em]
\toprule\\[-1.5em]
\textbf{Variable} & \textbf{Generating Mechanism} & \textbf{Notes}\\ \\[-1.5em]
\midrule\\[-1.2em]
$C_1$ & $\sim \mathcal{N}(0,1)$ & Measured covariate\\[2pt]
$C_2$ & $\sim \mathrm{Bernoulli}(0.4)$ & Measured covariate\\[3pt]
$A(a)$ & $\sim\mathrm{Bernoulli}(p_V)$ where $p_V = \left\{
\begin{aligned}
& 0 \quad \text{if } C_1 < -1,\\
& \text{logit}^{-1}(-1 + 0.5\,C_1 + 0.3\,C_2)  \quad \text{otherwise}\\
\end{aligned}
\right.$
 & Treatment status\\[8pt]
 $A(b)$ & $\sim\mathrm{Bernoulli}(p_V)$ where $p_V = \text{logit}^{-1}(-1 + 0.5\,C_1 + 0.3\,C_2) $
 & Treatment status\\[3pt]
$I_0$ & $\sim \mathrm{Bernoulli}\{logit(p)=-2.5 - 0.9\,C_1-1.2\,C_2\}$ &  Non-target pathogen\\[3pt]
$I_1$ & \makecell[l]{$\sim \mathrm{Bernoulli}\{logit(p)=-3.6+1.2\,C_1+1.2\,C_2-0.9\,A$\}} & Target pathogen\\[3pt]
$W_0$ & $\sim \mathrm{Bernoulli}\{logit(p)=0.6 +\,C_1\}$ for $I_0=1$ & Unobserved symptom\\[3pt]
$W_1$ & \makecell[l]{$\sim \mathrm{Bernoulli}\{logit(p)=-1.8+2\,C_1-0.8\,A\}$  for $I_1=1$ }& Unobserved symptom\\[3pt]
$W$ & $=max\{W_0, W_1)$ & Observed symptom\\[3pt]
$H$ & $\sim \mathrm{Bernoulli}\{logit(p)=1-0.5\,C_1\}$ for $W=1$  & Hospitalization\\
\bottomrule
\end{tabular}}
\end{table}

\subsection{Radius definition and adjustment}
Propensity score matching conventionally utilizes the logit of the estimated propensity score (PS) \cite{RosenbaumRubin1985}, $\Pi_i= \text{logit}(\pi_i)$, as the primary distance metric. The logit transformation is preferred because it linearizes the distance metric by stretching the boundaries near 0 and 1, ensuring a more uniform radius size across the entire propensity score distribution. The standard recommendation for the radius width ($\rho$) is set to $\rho = 0.2\times \sigma(\Pi)$ where $\sigma(\Pi)$ is the standard deviation of the logit-transformed propensity scores \cite{RosenbaumRubin1985, austin2011optimal}. 

In our two-stage propensity score matching, we defined two distinct radii based on the standard deviations of the logit transformed propensity scores: $\rho_1= d_1\times \hat{\sigma}_{\Pi^{ctrl}}, \rho_2= d_2\times \hat{\sigma}_{\Pi^{untreat}} $ where $\hat{\sigma}_{\Pi^{ctrl}} = \hat{\sigma}\{\Pi_i: Y_i=0\}, \hat{\sigma}_{\Pi^{untreat}} = \hat{\sigma}\{\Pi_i: A_i=0\}$. In a representative simulation ($N$= 5,000), these standard deviations were notably higher than those observed in the scenarios without partial non-overlap: for example, $\hat{\sigma}_{\Pi^{ctrl}}=6.58$ and $ \hat{\sigma}_{\Pi^{untreat}}=7.12$ for the CC sutdy, and 
 8.35 and 7.83, respectively, for the TND study.
The large standard deviation values reflect a highly dispersed logit distribution, driven by the extreme values produced by the one-sided positivity violation ($C_1<-1$). In such scenarios, using the conventional radius width of 0.2 would result in an excessively permissive matching threshold, potentially pairing individuals with highly dissimilar covariate profiles and undermining the goal of the matching procedure.

To identify the optimal radius width that maintains sufficient covariate balance, we evaluated a sequence of progressively tighter radius proportional constants. For the CC study, we set $d_1=(0.05, 0.03, 0.01, 0.005)$ and $d_2=(0.05, 0.03, 0.02, 0.01)$. At the upper end of this range, $d=0.05$, the resulting radius is approximately $0.33$ on the logit scale in the first stage. This tolerance would allow, for instance, an individual with a logit PS of 0 ($\pi=0.5$) to be matched to another with a logit PS of 0.33 ($\pi\approx 0.58$), representing a relatively wide radius. In contrast, if we set $d=0.01$, the implied radius is approximately 0.07 on the logit scale, allowing matches only between individuals with nearly identical propensity scores (e.g., $\pi = 0.50$ versus $\pi \approx 0.52$), thereby prioritizing covariate balance at the cost of reduced matching flexibility. For the TND study, we used $d_1=d_2=(0.05, 0.03, 0.02, 0.01)$. Using these scales, we conducted simulations under sample sizes ($N$=5,000), and presented the covariate balances and distribution overlap of logit propensity scores before and after matching.

\subsection{Diagnostic checking}

\subsubsection{Diagnostic checking in the CC simulation scenario (b) and (c)}
Figure \ref{fig: ps_cc_em_mRR_iptw} presents the propensity score distributions for mRR weighting in a single CC dataset of size $N=5,000$ from simulation scenario (b). The top row shows decent overlap before weighting and excellent overlap between untreated and treated controls after weighting. The bottom row suggests that the propensity score models, fit with control data, are extrapolating to predict values for the cases (little overlap of the blue cases above 0.75 and 0.85 on the right-hand side of each of the left and right bottom histograms, respectively). 
Figure \ref{fig: cc3_ps} presents the propensity score distributions for mRR weighting in  CC simulation scenario (c). Similarly, IPTW improved overlap between untreated and treated controls. Some propensity score extrapolation may again be necessary in this scenario.  

\begin{figure}[htbp]
 \centering
\includegraphics[width=12cm,height=8.5cm]{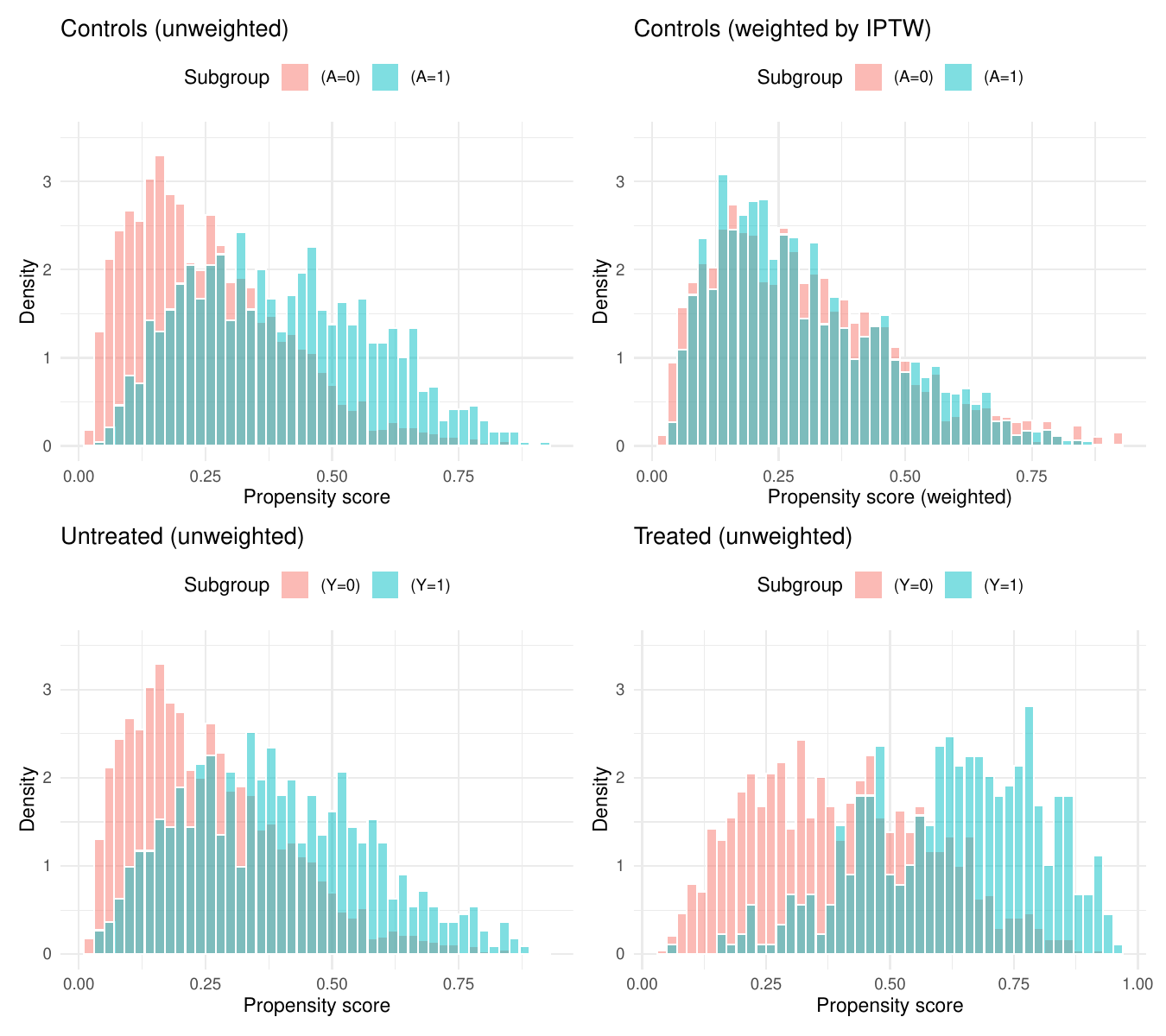}
 \caption{\small Diagnostic checks of the propensity score overlap for mRR weighting  in a single CC dataset of size $N=5,000$ in simulation scenario (b): comparison between the treated controls and untreated controls before and after IPTW weighting, between untreated cases and untreated controls, and between treated cases and treated controls. }
 \label{fig: ps_cc_em_mRR_iptw}
 \end{figure}

\begin{figure}[htbp]
 \centering
 \includegraphics[width=12cm,height=8.5cm]{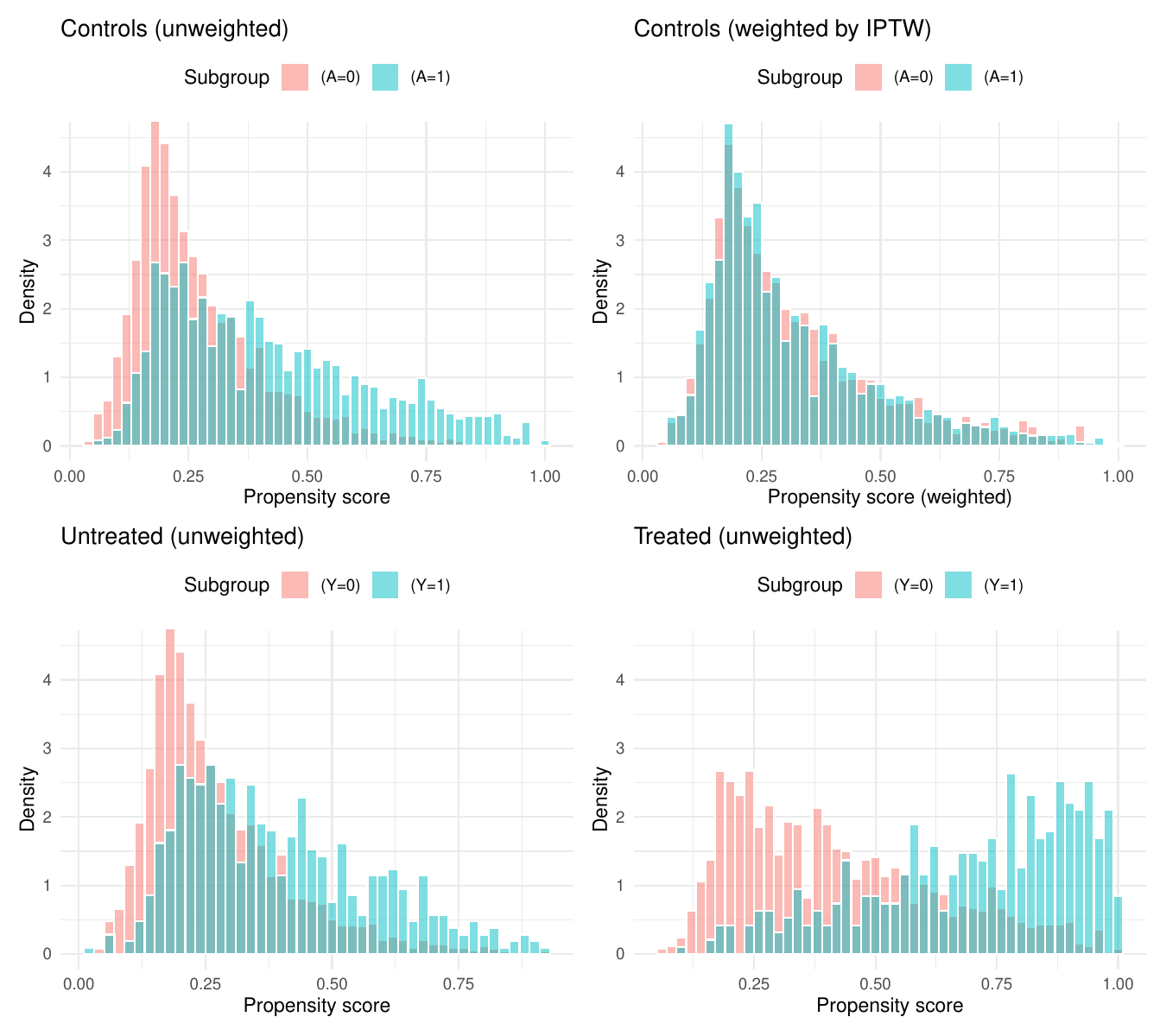}
 \caption{\small Diagnostic checks of the propensity score overlap for mRR weighting in a single CC dataset of size $N=5,000$ in simulation scenario (c): comparison between the treated controls and untreated controls before and after weighting for IPTW, between untreated controls and untreated cases before, and between treated cases and treated controls.}
 \label{fig: cc3_ps}
 \end{figure}

\clearpage

\subsubsection{Diagnostic checking in the TND simulation study}

Table \ref{tab:bc_tnd} summarizes the two-stage covariate balance checking  for IPTW and matching estimating the mRRT in a single TND dataset of $N=5,000$    using four different two-stage matching radii with the radius scaling parameters, $d_1=d_2\in \{0.05, 0.03, 0.02, 0.01$\}. Before matching,  both stages showed limited overlap and substantial imbalance in covariates, with large SMD particularly for $C_1$. Applying radius matching markedly improved alignment between the groups, both in the first stage, where matching was performed in the controls and in the second stage, where matching was performed between untreated cases and controls. 
Figure \ref{fig: tnd_ps} presents the propensity score distribution before and after IPTW and the two-stage matching weights. In the first stage (top row), after weighting by IPTW or matching, the two treatment control groups became more aligned in their distributions. In the second stage (bottom row), prior to matching, the two outcome untreated groups exhibited distinct distributions, with untreated controls concentrated above 0.2. Then the two groups showed substantially better overlap following the second stage matching. Moreover, the right tail of the untreated cases evident in the pre-matching panel (bottom left) is absent in the matched panel because these individuals had no eligible matches and therefore received zero weights.
Figure \ref{fig: ps_tnd_mRR_iptw} presents the propensity score distributions in controls before and after IPTW weighting, and the distributions in untreated and treated individuals for the mRR estimation. IPTW weighting resolved the small differences between treated and untreated in the control group (top row). The bottom row suggests that the propensity score models, fit with control data, are extrapolating to predict values for the cases (little overlap of the blue cases above 0.5 on the right-hand side of each of the bottom histograms).

\begin{table}[H]
\centering\small\captionsetup{width=.5\textwidth}
\caption{\footnotesize Diagnostic checks of the covariate balance checking for the mRRT weighting in a single simulated TND dataset: comparison among controls for IPTW and matching, and among untreated individuals for matching across matching radii. Variables are summarized as mean and standard deviation (mean (SD)); SMD: standardized mean difference. }\label{tab:bc_tnd}\vspace{2mm}
\resizebox{0.8\width}{0.28\textheight}{%
\begin{tabular}{lllccc}
  \toprule
 \multicolumn{6}{l}{ \emph{First stage: matching among the controls}} \vspace{2mm} \\ 
& $d_1$ & Variable & $(A=0, Y=0)$ & $(A=1, Y=0)$ & SMD \\ 
 \cmidrule(lr){2-6}
\multirow{3}{*}{Before weighting}& \multirow{3}{*}{-}& Size &  2427 &  589 &  \\ 
 & &  $C_1$  & -0.612 (0.855) & 0.021 (0.659) &  0.829 \\ 
&  &  $C_2$  & 0.174 (0.379) & 0.199 (0.399) &  0.064 \\  \cmidrule(lr){2-6}
\multirow{2}{*}{ IPTW} & \multirow{2}{*}{-} &  $C_1$  & 0.015 (0.630) & 0.021 (0.659)  &  0.008\\ 
&  &  $C_2$  & 0.198 (0.398)&  0.199 (0.399)&   0.002\\ \cmidrule(lr){2-6}
\multirow{12}{*}{Matching} & \multirow{3}{*}{ $0.05$} & Size & 1581 & 589 &  \\ 
 &  &  $C_1$  & -0.062 (0.57) & 0.021 (0.659) & 0.135 \\ 
&  &  $C_2$  & 0.177 (0.382) & 0.199 (0.399) & 0.055 \\  \cmidrule(lr){2-6}
& \multirow{3}{*}{ $0.03$} & Size & 1581 & 589 &  \\ 
 &  &  $C_1$  & -0.018 (0.605) & 0.021 (0.659) & 0.062 \\ 
 &  &  $C_2$  & 0.185 (0.389) & 0.199 (0.399) & 0.034 \\ \cmidrule(lr){2-6}
& \multirow{3}{*}{ $0.02$} & Size & 1581 & 589 &  \\ 
 &   &  $C_1$  & 0.003 (0.628) & 0.021 (0.659) & 0.028 \\ 
&  &  $C_2$  & 0.190 (0.393) & 0.199 (0.399) & 0.021 \\ \cmidrule(lr){2-6}
& \multirow{3}{*}{ $0.01$}& Size & 1581 & 589 &  \\ 
&   &  $C_1$  &  0.019 (0.65) & 0.021 (0.659) & 0.003 \\ 
 & &  $C_2$  & 0.193 (0.395) & 0.199 (0.399) & 0.014 \\ 
   \midrule
  \multicolumn{6}{l}{\emph{Second stage: donor matching (independent of $d_1$) between untreated cases and controls}}\vspace{2mm}\\ 
& $d_2$ & Variable &  $(A=0, Y=0)$ & $(A=0, Y=1)$ & SMD\\ 
 \cmidrule(lr){2-6}
\multirow{3}{*}{Before weighting} & \multirow{3}{*}{-}& Size &  2427 & 1466 &  \\ 
 &  &  $C_1$  & -0.612 (0.855) & 1.215 (0.760) &  2.254 \\ 
&  &  $C_2$   &  0.174 (0.380) & 0.601 (0.490) &  0.975 \\ \cmidrule(lr){2-6}
\multirow{12}{*}{Matching} 
&  \multirow{3}{*}{$0.05$} & Size & 2407 & 1457 &  \\ 
 &   &  $C_1$  & 0.8 (0.699) & 1.199 (0.747) & 0.553 \\ 
 & &  $C_2$  & 0.419 (0.493) & 0.598 (0.49) & 0.365 \\  \cmidrule(lr){2-6}
&  \multirow{3}{*}{$0.03$} & Size & 2377 & 1431 &  \\ 
 &  &  $C_1$  & 1.004 (0.69) & 1.169 (0.717) & 0.235 \\ 
 & &  $C_2$  & 0.493 (0.5) & 0.595 (0.491) & 0.208 \\  \cmidrule(lr){2-6}
& \multirow{3}{*}{$0.02$} & Size & 2345 & 1416 &  \\ 
 &  &  $C_1$  & 1.088 (0.682)&  1.154 (0.705) &  0.095 \\ 
 & &  $C_2$   & 0.524 (0.499) &  0.593 (0.491)  &  0.140   \\  \cmidrule(lr){2-6}
& \multirow{3}{*}{$0.01$} & Size & 2270 & 1387 &  \\ 
 &  &  $C_1$  & 1.132 (0.67) & 1.13 (0.691) & 0.004 \\ 
 &  &  $C_2$  &0.545 (0.498) & 0.587 (0.492) & 0.085 \\ 
   \bottomrule
\end{tabular}}
\end{table}

\begin{figure}[htbp]
 \centering
 \includegraphics[width=14cm,height= 10cm]{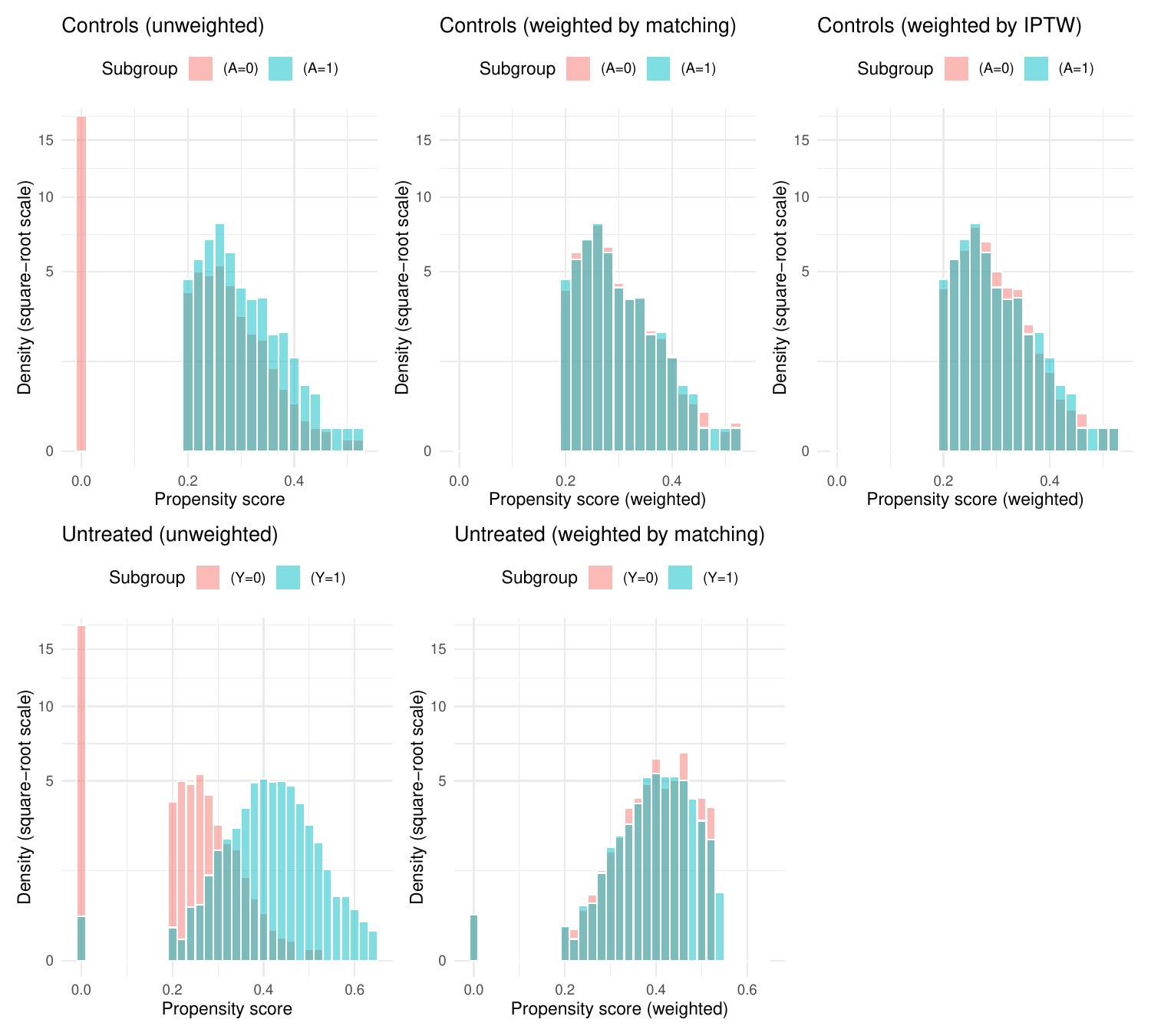}
 \caption{\small Diagnostic checks of the propensity score overlap for mRRT weighting in a single simulated TND dataset: comparison between the treated controls and untreated controls before and after weighting for IPTW and matching  ($d_1=0.01$ and $d_2=0.01$), and between untreated controls and untreated cases before and after matching. 
 }
 \label{fig: tnd_ps}
 \end{figure}

\begin{figure}[htbp]
 \centering
\includegraphics[width=12cm,height= 10cm]{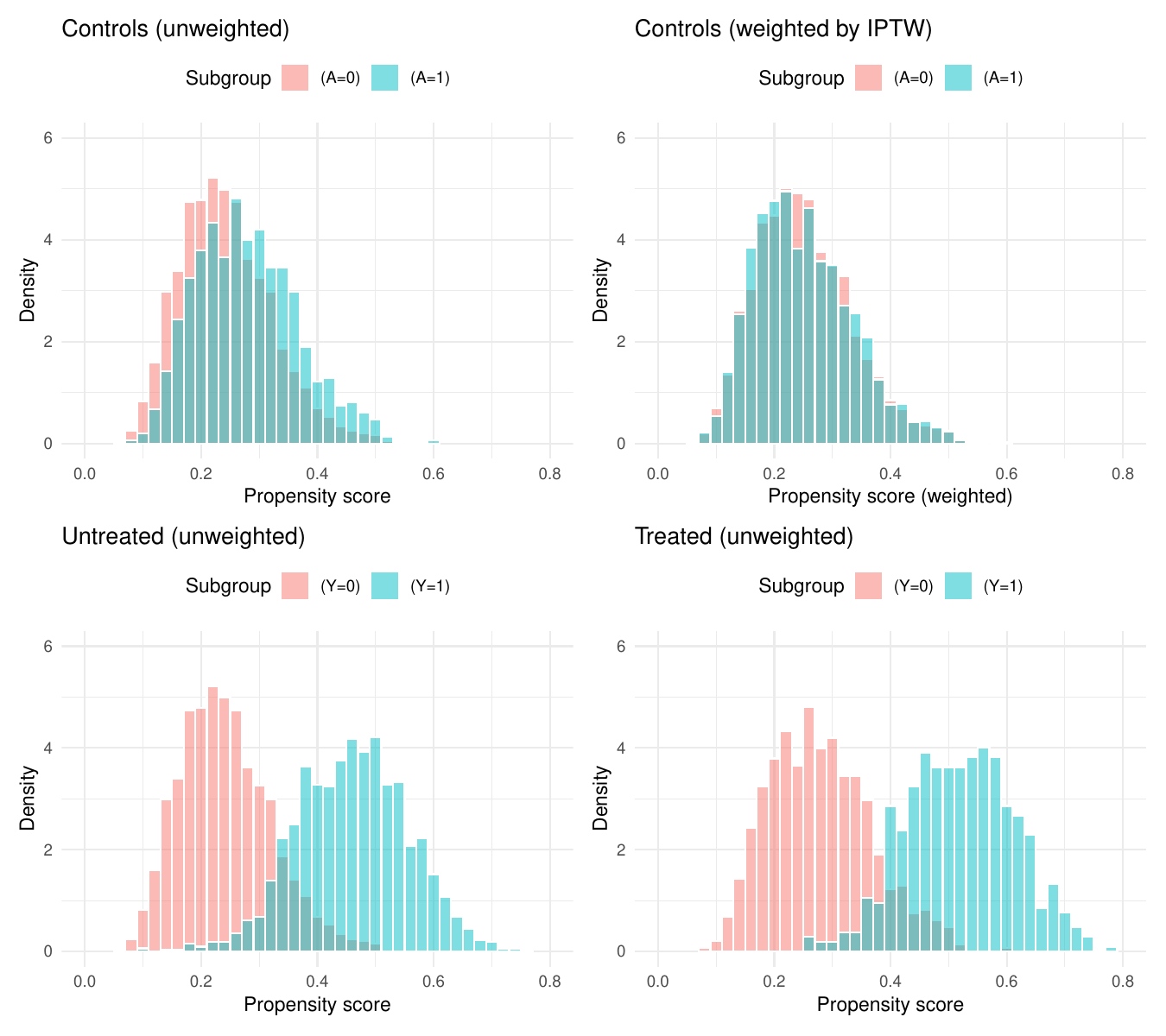}
 \caption{\small Diagnostic checks of the propensity score overlap for mRR weighting in a TND simulation study: comparison between the treated controls and untreated controls before and after IPTW weighting, between untreated cases and untreated controls, and between treated cases and treated controls. }
 \label{fig: ps_tnd_mRR_iptw}
 \end{figure}

\subsection{Results for simulations}

\subsubsection{Results for CC simulation scenario (c)}

Table \ref{tab:est_CC2} presents the estimates from CC simulation scenario (c), which generates CC data with a nonlinear treatment data-generating model. In implementing the IPTW and doubly robust one-step (OS) estimators, we specified propensity score models including only the main effects of $C_1$, $C_2$, and the outcome models including main terms of $C_1$, $C_2$ and $A$. IPTW\_SL and OS\_SL denote estimators in which nuisance functions were estimated using machine learning via Super Learner (SL.nnls, SL.glmnet, SL.hal9001). Standard error of IPTW\_SL was estimated using a double robust sandwich estimator, while standard error of OS\_SL was estimated based on the efficient influence function. We truncated the estimated probabilities $\hat{\mu_a}, \hat{\pi}_a$ for $a\in(1,0)$ to the interval [0.001, 0.999]. 

\begin{table}[ht]
\centering\small
\captionsetup{width=0.52\textwidth, font=footnotesize, skip=2pt}
\caption{Estimates in case--control simulation scenario (c)  (500 draws of $N=5,000$).  MC SE: Monte Carlo standard error; Coverage Rate: \% of confidence intervals that included the true value.}\label{tab:est_CC2}
\resizebox{\width}{!}{%
\begin{tabular}{lrrrrc}
  \toprule
Method & Estimate & Bias & MC SE &  SE & Coverage Rate (\%)\\  \midrule
  \multicolumn{6}{l}{\emph{Estimation of mRR (true value: 0.773)}}\\ 
IPTW  & 0.855  & 0.082 & 0.061 & 0.069 & 81.8 \\  
IPTW\_SL  & 0.820  & 0.047 & 0.072 &  0.075 & 89.6 \\ 
OS  &  0.735  & 0.038 & 0.069 & 0.063 &  85.8 \\
OS\_SL  &  0.764 & 0.009 & 0.113 & 0.105 & 91.8  \\
   \bottomrule
\end{tabular}}
\captionsetup{font=footnotesize}
\end{table}

\subsubsection{Results for the TND simulation study}
Table \ref{tab:est_TND} compares the performance of IPTW and the two-stage radius matching for estimation of the mRRT with different choices of the two-stage matching radius multipliers, as well as the performance of IPTW, OS and logistic regression estimators for estimation of the mRR based on 300 replications. For the OS estimator, we applied a correctly specified logistic regression for the outcome and propensity score models, and bounded between [0.01, 0.99] for $\mu_1$ and $\mu_0$.

\begin{table}[ht]
\centering\small\captionsetup{width=.60\textwidth}
\caption{\small Estimates in TND simulation study (300 draws of $N=5000$).  MC SE: Monte Carlo standard error; BS SE: Mean bootstrap-estimated standard error; Coverage Rate: \% of confidence intervals that included the true value}\label{tab:est_TND}\vspace{2mm}
 \resizebox{1\width}{0.21\textheight}{%
\begin{tabular}{lrrrrrrc}
  \toprule
Method & $d_1$ & $d_2$ & $\hat{\psi}$ & Bias & SE\_MC & SE\_bs & Coverage Rate (\%)\\  \midrule
  \multicolumn{8}{l}{\emph{For the mRRT estimation (true value: 0.392),}}\vspace{3mm}\\ 
IPTW  &- & - & 0.358 & 0.033 & 0.052 & 0.054 & 89.67 \\  \cmidrule(lr){2-8}
\multirow{16}{*}{Matching} & 
\multirow{4}{*}{0.05} &  0.05 & 0.632 & 0.240 & 0.116 & 0.097 & 23.67 \\ 
 &    & 0.03 & 0.636 & 0.244 & 0.132 & 0.111 & 29.67 \\ 
 &  & 0.02 & 0.647 & 0.256 & 0.139 & 0.118 & 29.33 \\ 
 &  & 0.01 & 0.677 & 0.285 & 0.141 & 0.123 & 19.33 \\  \cmidrule(lr){2-8}
 & \multirow{4}{*}{0.03} &  0.05 & 0.503 & 0.111 & 0.099 & 0.088 & 67.33 \\ 
 & & 0.03 & 0.485 & 0.093 & 0.106 & 0.099 & 80.33 \\ 
 &  & 0.02 & 0.487 & 0.095 & 0.109 & 0.105 & 82.00 \\ 
 &  & 0.01 & 0.450 & 0.059 & 0.098 & 0.103 & 91.67 \\   \cmidrule(lr){2-8}
 & \multirow{4}{*}{0.02}  & 0.05 & 0.458 & 0.066 & 0.091 & 0.085 & 85.33 \\ 
 & &0.03 & 0.433 & 0.042 & 0.096 & 0.095 & 91.33 \\ 
  & &0.02 & 0.432 & 0.040 & 0.098 & 0.100 & 92.00 \\ 
 & & 0.01 & 0.450 & 0.059 & 0.098 & 0.103 & 91.67 \\ \cmidrule(lr){2-8}
 & \multirow{4}{*}{0.01} & 0.05 & 0.431 & 0.039 & 0.086 & 0.083 & 90.67 \\ 
 & & 0.03 & 0.403 & 0.011 & 0.089 & 0.092 & 94.33 \\ 
  &  & 0.02 & 0.398 & 0.006 & 0.091 & 0.097 & 95.33 \\ 
 &  &0.01 & 0.409 & 0.017 & 0.092 & 0.100 & 96.00 \\ \midrule
 \multicolumn{8}{l}{\emph{For the mRR estimation (true value: 0.371),}}\vspace{3mm}\\ 
IPTW  & -& - & 0.340  & 0.032 & 0.035 & 0.042&  92.67 \\  
OS   & -& - & 0.323  & 0.047 & 0.081 & 0.082&  86.67 \\
Logistic & -& - & 0.267  & 0.104 & 0.031 & 0.034&  22.67 \\
   \bottomrule
\end{tabular}}
\end{table}

For the mRRT estimation, the IPTW exhibited moderate bias (0.033) and coverage (89.67\%) below the nominal level. In contrast, the performance of the matching estimators was highly sensitive to the choices of $d_1$. When large radii were used in the first stage ($d_1 =0.05$), the matching estimator was substantially biased (exceeding 0.24) and had severe undercoverage (below 30\%), indicating inadequate covariate balance due to overly permissive matching for the propensity scores. As the matching radii were tightened, bias generally decreased  and coverage improved markedly. In particular, combinations with small radii in both stages (e.g,. $d_1\in \{0.02, 0.01\}$ and $d_1 \in \{0.03, 0.02, 0.01\}$) achieved near-zero bias, as low as 0.006. The Monte Carlo and bootstrap standard errors were well-aligned, and coverage rates were close to or exceeding 95\%. These results highlight the importance of carefully tuning both stages of the matching procedures and show that, with appropriately chosen radii, two-stage matching can outperform IPTW in terms of bias and coverage. For mRR estimation, IPTW exhibited lower bias and higher coverage rates than the OS estimator, which showed larger bias and undercoverage relative to the true mRR value. This discrepancy is likely attributable to the use of upper bounds (0.98) on the estimated conditional outcome probabilities, as some values of $\hat{\mu}_a(\boldsymbol{c}_i)$ approached one, which can lead to instability, as discussed in Section 2.4.2 of the manuscript. Logistic regression performed poorly despite the absence of effect modification.

\subsubsection{Demonstration of double robustness of the OS estimator}
Table \ref{tab:dr_CC} presents the doubly robustness property of the OS estimator in CC simulation scenario (b). We can see that when either the propensity score or outcome models were correctly specified, the OS estimator was unbiased. In contrast, when the propensity score model was null (i.e. only included an intercept term in the logistic regression), IPTW was highly biased. 

\begin{table}[ht!]
\centering\small
\captionsetup{width=.70\textwidth}
\caption{\small Demonstration of double robustness of the OS estimator in case--control simulation scenario (b). Bias, Monte Carlo standard errors (MC SE), mean bootstrap-estimated standard errors (BS SE), and coverage rates.}\label{tab:dr_CC}\vspace{1mm}
\resizebox{0.7\textwidth}{!}{%
\begin{tabular}{lrrrrrc}
  \toprule
Scenario & Method & $\hat{\psi}_{mRR}$ & Bias & MC SE & BS SE & Coverage Rate (\%)\\\midrule
  \multicolumn{7}{l}{\emph{True value: ${\psi}_{mRR}=0.770$}}\vspace{1mm}\\ 
\multirow{2}{*}{Both true}  & IPTW & 0.777  & 0.006 & 0.066 & 0.066 & 93.60 \\  
 & OS &  0.786 & 0.015 & 0.073 & 0.078 & 94.60  \\\cmidrule(lr){2-7}
 \multirow{2}{*}{Only PS true}  & IPTW & 0.777  & 0.006 & 0.066 & 0.066 & 93.60 \\  
 & OS &  0.779 & 0.008 & 0.066 & 0.067 & 94.00  \\\cmidrule(lr){2-7}
 \multirow{2}{*}{Only outcome true}  & IPTW & 1.827  & 1.056 & 0.126 & 0.133 & 0 \\  
 & OS &  0.783 & 0.012 & 0.081 & 0.081 & 94.80  \\\cmidrule(lr){2-7}
 \multirow{2}{*}{Both null}  & IPTW & 1.827  & 1.056 & 0.126 & 0.133 & 0 \\  
 & OS &  1.827  & 1.056 & 0.126 & 0.133 & 0   \\
   \bottomrule
   \noalign{\vskip 1mm}
   \multicolumn{7}{l}{\footnotesize\makecell[l]{Both true: both propensity score model and outcome model were correctly specified; Only PS true: the \\propensity score model was correct but the outcome model was null; Only outcome true: the outcome\\ model was correct but the propensity score model was null; Both null: both models were null.}}
\end{tabular}}
\end{table}

\section{Application}

\subsection{PROVAQ data}
Table \ref{tab:descriptive_app} presents the characteristics of the study population by case status in the PROVAQ study. Continuous variables are presented as mean (standard deviation); Categorical variables are presented as size (percentage), where percentages are calculated within each column.

\begin{table}[ht]
\centering\small\captionsetup{width=.56\textwidth}
 \caption{\small Descriptive characteristics of the study population by case status in the PROVAQ study. }
 \label{tab:descriptive_app}\vspace{1mm}
 \resizebox{1\width}{0.22\textheight}{%
\begin{tabular}{lrrr}
  \toprule
 Variable & Overall & Controls & Cases \\ 
  \midrule
Size  &  1,406  &   908  &   498  \\ 
  Age (mean (SD))  & 57.87 (12.64)  & 58.19 (12.62)  & 57.28 (12.66)  \\ 
  Education (\%):  &    &     &     \\ 
  \hspace{3mm}   Less HS  &   136 (9.7)   &    83 (9.1)   &    53 (10.6)   \\ 
  \hspace{3mm}   HS  &   339 (24.1)   &   199 (21.9)   &   140 (28.1)   \\ 
   \hspace{3mm}  College  &   422 (30.0)   &   278 (30.6)   &   144 (28.9)   \\ 
  \hspace{3mm}   Undergraduate  &   356 (25.3)   &   244 (26.9)   &   112 (22.5)   \\ 
  \hspace{3mm}   Graduate  &   153 (10.9)   &   104 (11.5)   &    49 ( 9.8)   \\ 
  BMI (mean (SD))  & 26.21 (5.85)  & 25.92 (5.54)  & 26.73 (6.36)  \\ 
  Endometriosis: ever (\%)  &    95 ( 6.8)   &    50 ( 5.5)   &    45 ( 9.0)   \\ 
  HRT: ever (\%)  &   450 (32.0)   &   282 (31.1)   &   168 (33.7)   \\ 
  Duration of OC use: (\%):  &    &     &     \\ 
  \hspace{3mm}   0  &   280 (19.9)   &   172 (18.9)   &   108 (21.7)   \\ 
  \hspace{3mm}   $>$0-2 years  &   253 (18.0)   &   158 (17.4)   &    95 (19.1)   \\ 
  \hspace{3mm}   2-10 years  &   533 (37.9)   &   336 (37.0)   &   197 (39.6)   \\ 
  \hspace{3mm}   10+ years  &   340 (24.2)   &   242 (26.7)   &    98 (19.7)   \\ 
  Ancestry (\%):  &    &     &     \\ 
  \hspace{3mm}   French Canadian  &   948 (67.4)   &   609 (67.1)   &   339 (68.1)   \\ 
   \hspace{3mm}  Other European  &   333 (23.7)   &   217 (23.9)   &   116 (23.3)   \\ 
   \hspace{3mm}  Other/mixed  &   125 ( 8.9)   &    82 ( 9.0)   &    43 ( 8.6)   \\ 
  Smoking (\%):  &    &     &     \\ 
   \hspace{3mm}  Never  &   628 (44.7)   &   426 (46.9)   &   202 (40.6)   \\ 
   \hspace{3mm}  $>$0 to 15 packyears  &   379 (27.0)   &   237 (26.1)   &   142 (28.5)   \\ 
   \hspace{3mm}  $>$15 packyears  &   399 (28.4)   &   245 (27.0)   &   154 (30.9)   \\ 
  Alcohol: ever (\%)  &  1,022 (72.7)   &   672 (74.0)   &   350 (70.3)   \\ 
  Aspirin: regular use (\%)  &   215 (15.3)   &   147 (16.2)   &    68 (13.7)   \\ 
  Otherthan aspirin*: yes (\%)  &   381 (27.1)   &   250 (27.5)   &   131 (26.3)   \\
   \bottomrule
\end{tabular}}
\captionsetup{font=footnotesize, justification=raggedright,
  singlelinecheck=false}
\caption*{\textit{Notes:} Otherthan aspirin means regular use of na-nsaids or acetaminophen. }
\end{table}

\subsection{Diagnostic checking}
Table \ref{tab:cov_bal_application} presents the covariate balance: SMD evaluation in controls and in untreated individuals before and after IPTW and matching weighting ($d_1=d_2=0.2$).
Figure \ref{fig: app_ps_mrrt} shows the propensity score overlap for mRRT weighting: comparison between the treated controls and untreated controls before and after weighting for IPTW and matching, and between untreated controls and untreated cases before and after matching. 
Figure \ref{fig: app_ps_mrr} presents the propensity score overlap for mRR weighting: comparison between the treated controls and untreated controls before and after IPTW weighting, between untreated cases and untreated controls, and between treated cases and treated controls. 
Figure \ref{fig: app_ps_ml} presents the propensity score overlap using IPTW with machine learning: comparison between the treated controls and untreated controls before and after IPTW weighting for mRRT and for mRR, between untreated cases and untreated controls, and between treated cases and treated controls. 
Table \ref{tab:cov_balds_application} presents the diagnostic checks of the covariate balance for matching under different values of $d_1$ and $d_2$ in the PROVAQ study. We selected ($d_1=d_2=0.2$) for the final analysis because this choice successfully matched all cases who did not regularly use aspirin during the second stage of matching and yielded the best overall covariate balance among the specifications examined.

\begin{table}[ht]
\centering\small\captionsetup{width=.86\textwidth}
 \caption{Diagnostic checks of the covariate balance: SMD evaluation in controls and in untreated individuals before and after IPTW and matching weighting ($d_1=d_2=0.2$) in the PROVAQ study. }
 \label{tab:cov_bal_application}\vspace{2mm}
  \resizebox{1\width}{0.23\textheight}{%
\begin{tabular}{lrrrrrrrr}
  \toprule
  \multirow{4}{*}{{Variable}} & \multicolumn{6}{c}{In controls} & \multicolumn{2}{c}{In untreated individuals} \\\cmidrule(lr){2-7}\cmidrule(lr){8-9}
   &  \multirow{2}{*}{{Unweighted}}  & \multicolumn{3}{c}{For mRRT} &  \multicolumn{2}{c}{For mRR} &  \multirow{2}{*}{{Unweighted}} & For mRRT\\ \cmidrule(lr){3-5}\cmidrule(lr){6-7}\cmidrule(lr){9-9}
    & & IPTW  & IPTW\_SL& Matching &  IPTW & IPTW\_SL &   & Matching  \\ 
  \midrule
Age & 0.934 & 0.019 & 0.078 & 0.008 & 0.040 & 0.191 & 0.049 & 0.103 \\  
Education:  &   &   &  &   &  & &   &  \\  
\hspace{3mm}Less HS & 0.355 & 0.015 & 0.094 & 0.040 & 0.008 & 0.080 & 0.059 & 0.030 \\ 
\hspace{3mm}HS & 0.073 & 0.005 & 0.137 & 0.019 & 0.026 & 0.022 & 0.146 & 0.135 \\ 
\hspace{3mm}College & 0.036 & 0.009 & 0.090 & 0.036 & 0.041 & 0.007 & 0.019 & 0.006 \\ 
\hspace{3mm}Undergraduate & 0.326 & 0.019 & 0.079 & 0.006 & 0.046 & 0.096 & 0.143 & 0.129 \\ 
\hspace{3mm}Graduate & 0.004 & 0.016 & 0.051 & 0.018 & 0.034 & 0.092 & 0.023 & 0.023 \\ 
BMI & 0.387 & 0.041 & 0.118 & 0.021 & 0.032 & 0.160 & 0.184 & 0.161 \\ 
Endometriosis: ever & 0.187 & 0.063 & 0.037 & 0.002 & 0.033 & 0.029 & 0.172 & 0.170 \\ 
HRT: ever & 0.331 & 0.035 & 0.004 & 0.002 & 0.043 & 0.051 & 0.058 & 0.032 \\ 
Duration of OC use:  & &  &  &  &  &  &   &  \\ 
\hspace{3mm}0 & 0.201 & 0.006 & 0.055 & 0.011 & 0.125 & 0.098 & 0.055 & 0.035 \\ 
\hspace{3mm}$>$0-2 years & 0.030 & 0.005 & 0.023 & 0.018 & 0.074 & 0.060 & 0.036 & 0.032 \\ 
\hspace{3mm}2-10 years & 0.060 & 0.004 & 0.127 & 0.009 & 0.005 & 0.030 & 0.067 & 0.067 \\ 
\hspace{3mm}10+ years & 0.298 & 0.007 & 0.117 & 0.017 & 0.049 & 0.185 & 0.157 & 0.136 \\ 
Ancestry: & &  &  &  &  &  &  &  \\ 
\hspace{3mm}French Canadian & 0.148 & 0.006 & 0.096 & 0.014 & 0.018 & 0.004 & 0.041 & 0.033 \\ 
\hspace{3mm}Other European & 0.100 & 0.009 & 0.123 & 0.002 & 0.048 & 0.005 & 0.031 & 0.028 \\ 
\hspace{3mm}Other/mixed & 0.097 & 0.027 & 0.034 & 0.022 & 0.111 & 0.014 & 0.022 & 0.013 \\ 
Smoking: &  &   &   &  &  &  &  &   \\ 
\hspace{3mm}Never & 0.032 & 0.018 & 0.003 & 0.003 & 0.000 & 0.025 & 0.159 & 0.144 \\ 
\hspace{3mm}$>$0 to 15 packyears & 0.140 & 0.013 & 0.055 & 0.018 & 0.092 & 0.016 & 0.070 & 0.068 \\ 
\hspace{3mm}$>$15 packyears & 0.167 & 0.031 & 0.046 & 0.013 & 0.097 & 0.044 & 0.105 & 0.090 \\ 
Alcohol: ever & 0.033 & 0.035 & 0.063 & 0.019 & 0.168 & 0.084 & 0.038 & 0.042 \\ 
Otherthan aspirin:yes & 0.394 & 0.043 & 0.086 & 0.018 & 0.016 & 0.106 & 0.007 & 0.015 \\ 
   \bottomrule
\end{tabular}}
\end{table}

\begin{figure}[htbp]
 \centering
 \includegraphics[width=14cm,height=10cm]{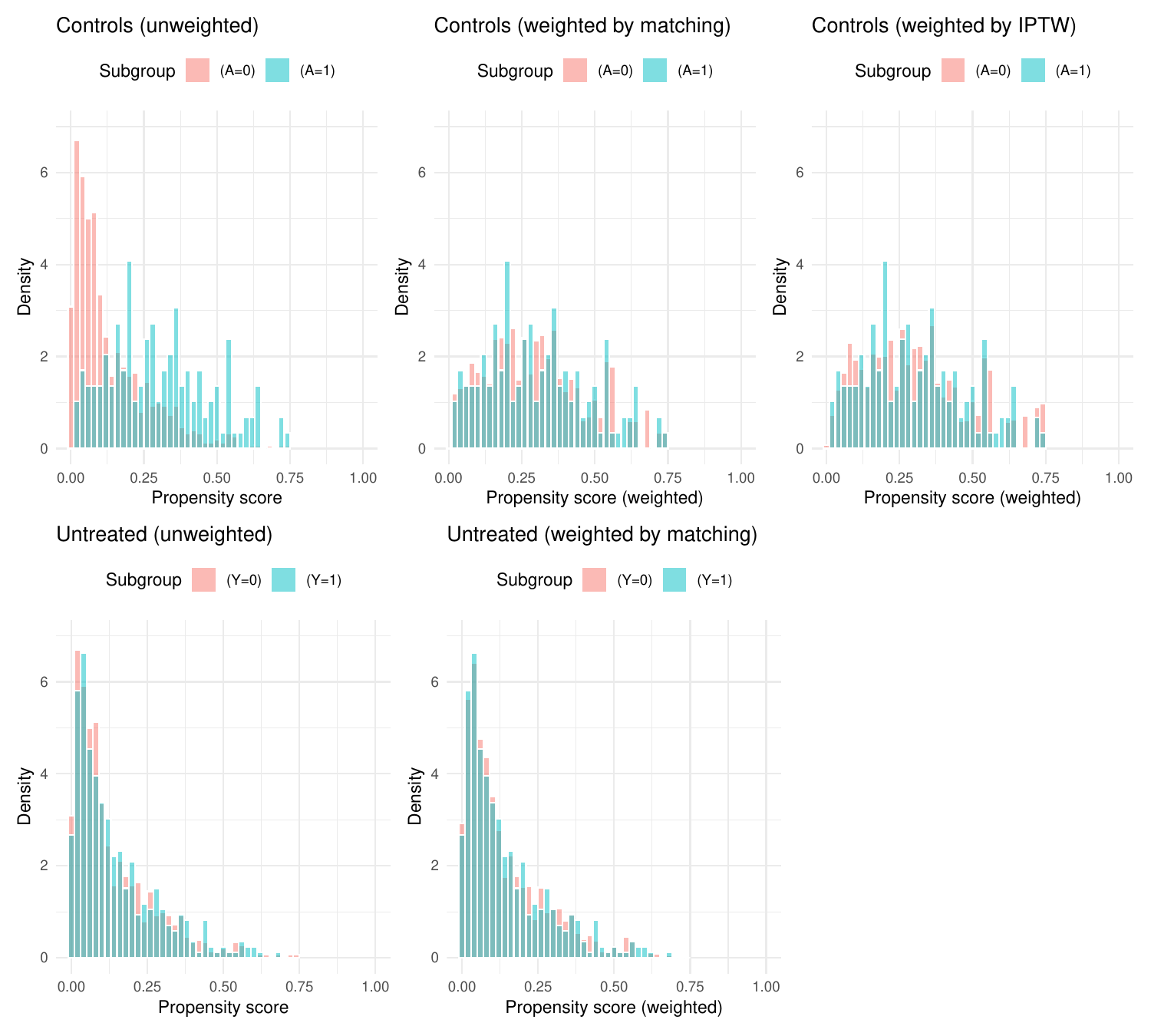}
 \caption{\small Diagnostic checks of the propensity score overlap for mRRT  in the PROVAQ study: comparison between the treated controls and untreated controls before and after weighting for IPTW and matching ($d_1=d_2=0.2$), and between untreated controls and untreated cases before and after matching. }
 \label{fig: app_ps_mrrt}
 \end{figure}

\begin{figure}[htbp]
 \centering
 \includegraphics[width=12cm,height=10cm]{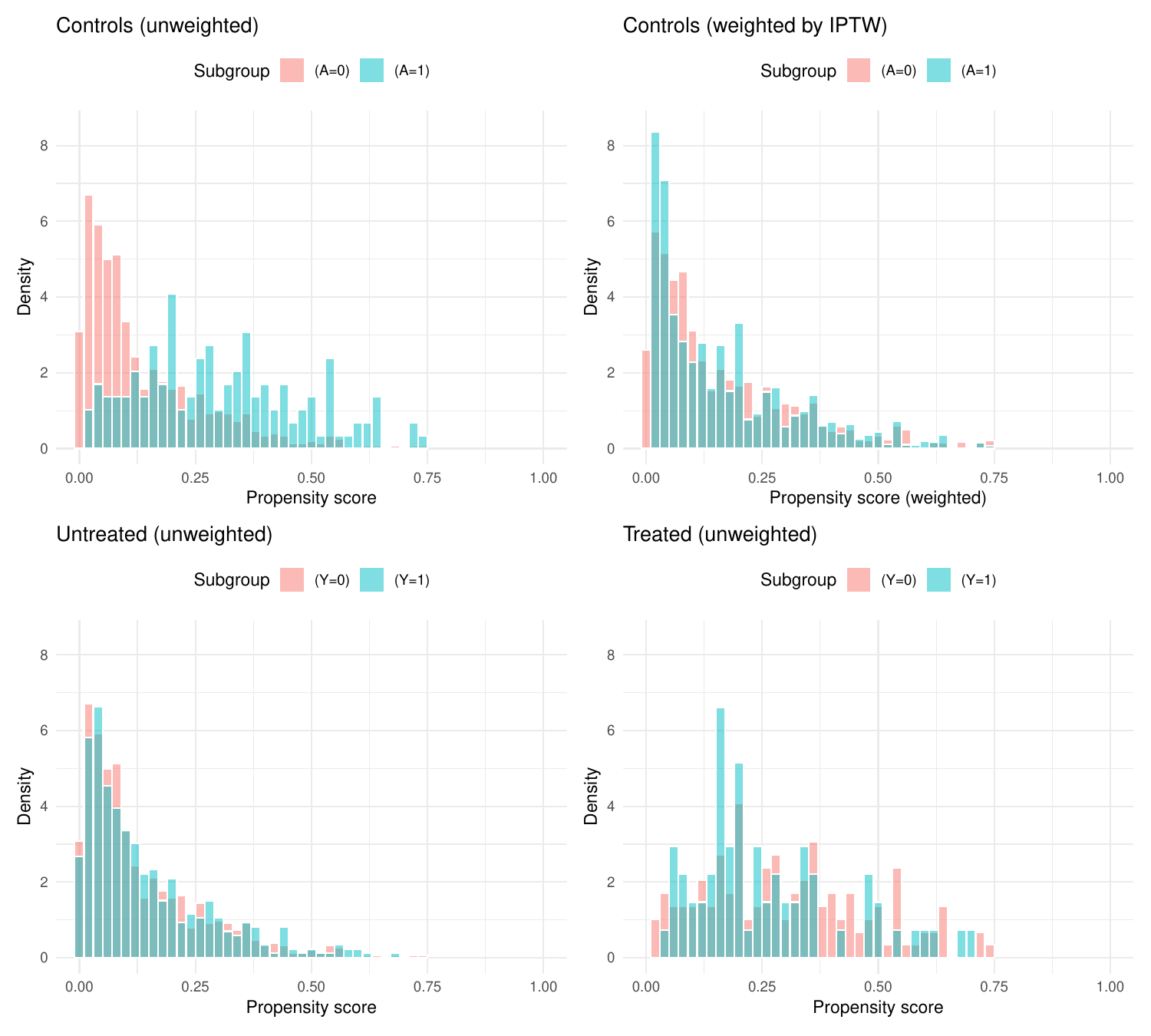}
 \caption{\small Diagnostic checks of the propensity score overlap for mRR in the PROVAQ study: comparison between the treated controls and untreated controls before and after IPTW weighting, between untreated cases and untreated controls, and between treated cases and treated controls. }
 \label{fig: app_ps_mrr}
 \end{figure}

\begin{figure}[htbp]
 \centering
 \includegraphics[width=14cm,height=10cm]{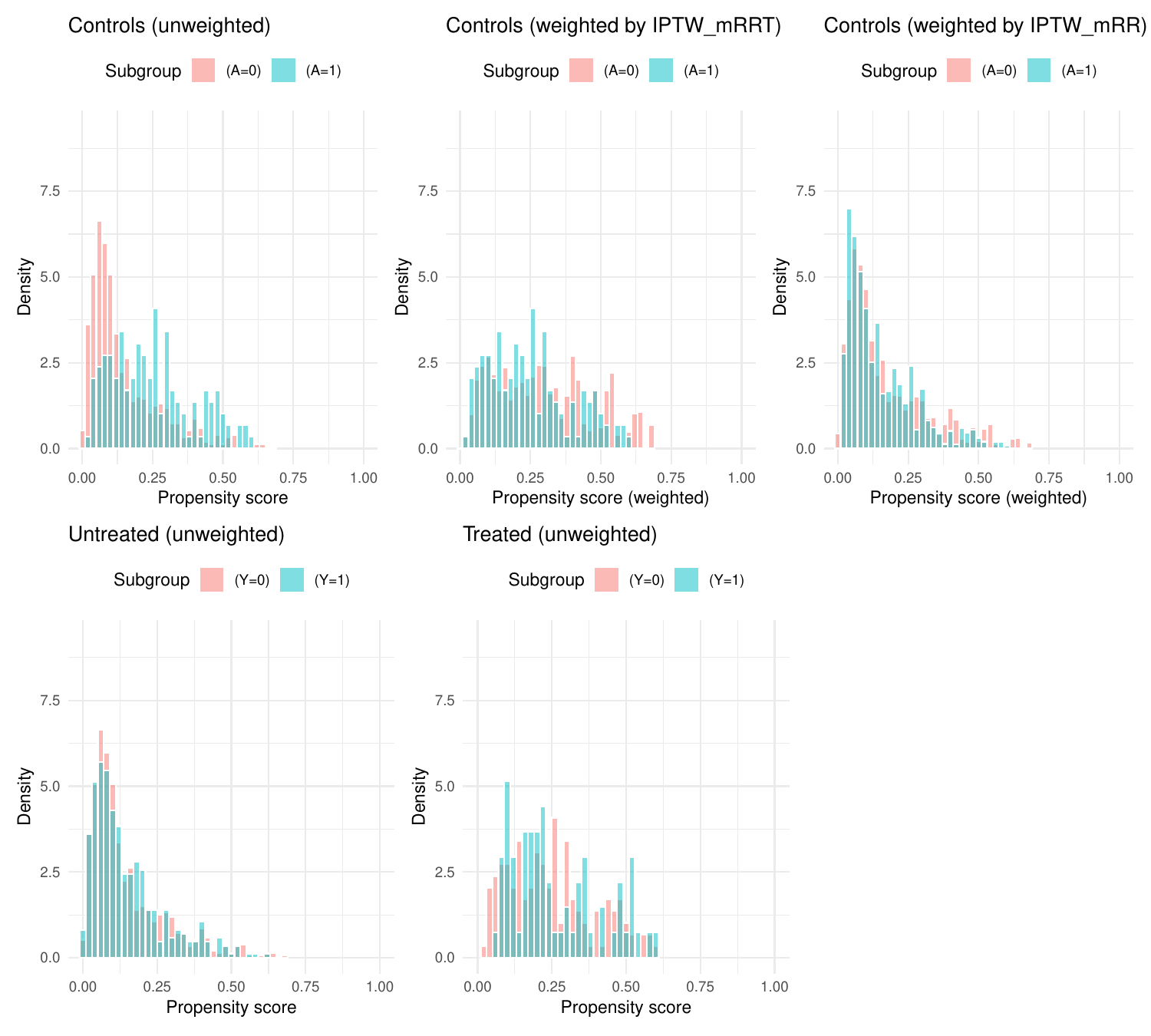}
 \caption{\small Diagnostic checks of the propensity score overlap with machine learning in the PROVAQ study: comparison between the treated controls and untreated controls before and after IPTW weighting for mRRT and for mRR, between untreated cases and untreated controls, and between treated cases and treated controls. }
 \label{fig: app_ps_ml}
 \end{figure}

\begin{table}[ht]
\centering\small\captionsetup{width=.99\textwidth}
 \caption{Diagnostic checks of the covariate balance for matching under different values of $d_1$ and $d_2$ in the PROVAQ study: SMD evaluation in controls and in untreated individuals before and after matching weighting. } \label{tab:cov_balds_application}\vspace{2mm}
  \resizebox{1\width}{0.23\textheight}{%
\begin{tabular}{lrrrrrrrrrrrr}
  \toprule
 \multirow{3}{*}{{Variable}}  &   \multicolumn{6}{c}{In controls}   &   \multicolumn{6}{c}{In untreated individuals }  \\ \cmidrule(lr){2-7}\cmidrule(lr){8-13}
   & \multirow{2}{*}{{Unweighted}}&   \multicolumn{5}{c}{$d_1$}& \multirow{2}{*}{{Unweighted}} &  \multicolumn{5}{c}{ $d_2$}\\ \cmidrule(lr){3-7}\cmidrule(lr){9-13}
   &   &  0.2 &  0.16 &  0.12 &  0.08 &  0.04 &   &  0.2 &  0.16 &  0.12 &  0.08 &  0.04 \\ 
  \midrule
Age & 0.934 & 0.008 & 0.005 & 0.007 & 0.026 & 0.000 & 0.049 & 0.103 & 0.105 & 0.111 & 0.112 & 0.117 \\ 
 Education: &   &   &  &   &  &  &  &  &  &  &  &  \\ 
  \hspace{3mm} Less HS & 0.355 & 0.040 & 0.051 & 0.064 & 0.066 & 0.080 & 0.059 & 0.030 & 0.026 & 0.020 & 0.027 & 0.024 \\ 
  \hspace{3mm} HS & 0.073 & 0.019 & 0.019 & 0.019 & 0.021 & 0.028 & 0.146 & 0.135 & 0.135 & 0.145 & 0.138 & 0.145 \\ 
  \hspace{3mm} College & 0.035 & 0.036 & 0.035 & 0.039 & 0.038 & 0.034 & 0.019 & 0.006 & 0.004 & 0.003 & 0.003 & 0.007 \\ 
  \hspace{3mm} Undergraduate & 0.326 & 0.006 & 0.015 & 0.022 & 0.027 & 0.032 & 0.143 & 0.129 & 0.130 & 0.138 & 0.142 & 0.124 \\ 
  \hspace{3mm} Graduate & 0.004 & 0.018 & 0.023 & 0.027 & 0.027 & 0.058 & 0.023 & 0.023 & 0.021 & 0.020 & 0.020 & 0.037 \\ 
  BMI & 0.387 & 0.021 & 0.018 & 0.034 & 0.046 & 0.037 & 0.184 & 0.161 & 0.162 & 0.165 & 0.164 & 0.162 \\ 
 Endometriosis: ever & 0.187 & 0.002 & 0.001 & 0.001 & 0.002 & 0.024 & 0.172 & 0.170 & 0.176 & 0.173 & 0.178 & 0.194 \\ 
  Everhrt: ever & 0.331 & 0.002 & 0.008 & 0.021 & 0.035 & 0.070 & 0.058 & 0.032 & 0.031 & 0.036 & 0.035 & 0.066 \\ 
  Duration of OC use:  &  &  &  &  &  &  &  &  &  &  &  &  \\ 
   \hspace{3mm} 0 & 0.201 & 0.011 & 0.007 & 0.019 & 0.018 & 0.005 & 0.055 & 0.035 & 0.034 & 0.036 & 0.028 & 0.042 \\ 
   \hspace{3mm} $>$0-2 years & 0.030 & 0.018 & 0.018 & 0.052 & 0.039 & 0.069 & 0.036 & 0.032 & 0.032 & 0.029 & 0.012 & 0.004 \\ 
  \hspace{3mm} 2-10 years & 0.060 & 0.008 & 0.006 & 0.009 & 0.009 & 0.010 & 0.067 & 0.067 & 0.068 & 0.069 & 0.069 & 0.075 \\ 
  \hspace{3mm} 10+ years & 0.298 & 0.017 & 0.018 & 0.019 & 0.029 & 0.049 & 0.157 & 0.136 & 0.135 & 0.136 & 0.115 & 0.119 \\ 
  Ancestry: &  &  &  &  &  &  &  &  &  &  &  &   \\ 
  \hspace{3mm} French Canadian & 0.148 & 0.014 & 0.021 & 0.026 & 0.029 & 0.050 & 0.041 & 0.033 & 0.039 & 0.050 & 0.051 & 0.061 \\ 
  \hspace{3mm} Other European & 0.100 & 0.002 & 0.009 & 0.014 & 0.014 & 0.038 & 0.031 & 0.028 & 0.031 & 0.041 & 0.040 & 0.040 \\ 
  \hspace{3mm} Other/mixed & 0.097 & 0.022 & 0.023 & 0.023 & 0.029 & 0.026 & 0.022 & 0.013 & 0.018 & 0.022 & 0.024 & 0.039 \\ 
  Smoking: &  &  &   &  & &   &  &   &   &   & &   \\ 
  \hspace{3mm} Never & 0.032 & 0.003 & 0.003 & 0.010 & 0.015 & 0.006 & 0.159 & 0.144 & 0.148 & 0.147 & 0.140 & 0.167 \\ 
  \hspace{3mm} $>$0 to 15 packyears & 0.140 & 0.018 & 0.015 & 0.013 & 0.001 & 0.031 & 0.070 & 0.068 & 0.068 & 0.068 & 0.061 & 0.076 \\ 
  \hspace{3mm} $>$15 packyears & 0.167 & 0.013 & 0.016 & 0.022 & 0.017 & 0.021 & 0.105 & 0.090 & 0.094 & 0.092 & 0.091 & 0.107 \\ 
  Alcohol: ever & 0.033 & 0.019 & 0.014 & 0.027 & 0.045 & 0.060 & 0.038 & 0.042 & 0.040 & 0.032 & 0.022 & 0.018 \\ 
  Otherthan aspirin: yes & 0.394 & 0.018 & 0.030 & 0.035 & 0.024 & 0.050 & 0.007 & 0.015 & 0.012 & 0.007 & 0.002 & 0.001 \\ \midrule
  Unmatched rate (\%) & -  &  0& 0 &  0&  0&  0 & - & 0 & 0 &  0.23 &  0.47 & 2.09  \\ 
   \bottomrule
\end{tabular}}
\end{table}

\subsection{Results for mRR estimation under three truncation intervals}
Table \ref{tab:est_app_trun} presents the estimated mRR using IPTW and one-step doubly robust estimator under three truncation intervals in the PROVAQ study. IPTW\_SL and OS\_SL denote estimators in which nuisance functions were estimated using machine learning via Super Learner (SL.glm, SL.gam, SL.glmnet, SL.randomForest), where the subscripts ${0.001}$, ${0.01}$, and  ${0.05}$ indicate the truncation intervals [0.001, 0.999], [0.01, 0.99], and [0.05, 0.95], respectively.

\begin{table}[H]
   \centering\small
\captionsetup{width=0.37\textwidth, font=footnotesize, skip=2pt}
\caption{Estimated mRR using IPTW and one-step (OS) doubly robust estimator under three truncation intervals in the PROVAQ study.  BS SE: Mean bootstrap-estimated standard error.}\label{tab:est_app_trun}
 \resizebox{0.95\width}{0.085\textheight}{%
\begin{tabular}{lrrc}
  \toprule
Method & Estimate & BS SE & 95\% CI\\  \midrule
IPTW  & 0.771 & 0.289 & [0.550, 1.627] \\
IPTW\_SL$_{0.001}$ & 0.716 & 0.125 & [0.467, 0.942] \\
IPTW\_SL$_{0.01}$ & 0.716 & 0.125 & [0.475, 0.953] \\
IPTW\_SL$_{0.05}$ & 0.714 & 0.116 & [0.460, 0.897] \\
OS  &  0.781 & 0.315 & [0.506, 1.634]  \\
OS\_SL$_{0.001}$ & 0.674 & 0.617 & [0.082, 2.782] \\
OS\_SL$_{0.01}$ & 0.674 & 0.347 & [0.238, 1.641] \\
OS\_SL$_{0.05}$ & 0.718 & 0.202 & [0.597, 1.404] \\
   \bottomrule
\end{tabular}}
\captionsetup{font=footnotesize}
\end{table}

\end{document}